\documentclass[12pt, draftclsnofoot, onecolumn]{IEEEtran}
\usepackage{amssymb,amsmath}
\usepackage{dsfont}
\usepackage{cite}
\usepackage{graphicx,subfigure}
\usepackage{psfrag}
\usepackage{url}
\usepackage[latin1]{inputenc}
\usepackage[absolute,overlay]{textpos}
\usepackage{tikz}
\usetikzlibrary{arrows,calc,decorations.markings}
\usetikzlibrary{topaths}
\usetikzlibrary{shapes,trees,snakes}
\usetikzlibrary{arrows}
\usetikzlibrary{shadows}
\usetikzlibrary{positioning}
\usetikzlibrary{matrix}
\usetikzlibrary{shapes.geometric}
\usetikzlibrary{decorations.pathmorphing}
\usepgflibrary{patterns}
\usetikzlibrary{calc}
\usetikzlibrary{fit}					
\usetikzlibrary{backgrounds}	

\usepackage{pgf}
\usetikzlibrary{arrows,automata}
\usepackage[latin1]{inputenc}
\usepackage{ragged2e}
\usepackage{algorithm}
\usepackage{epstopdf}
\usepackage{relsize}

\usepackage{algpseudocode}

\usepackage{soul}

\usepackage{amsthm}

\newtheorem{lemma}{Lemma}
\newtheorem{theorem}{Theorem}

\newtheorem{remark}{Remark}

\begin{document}

\title{Wireless Powered Communications with Finite Battery and Finite Blocklength}
\author{
	\IEEEauthorblockN{	Onel L. Alcaraz López, 
		Evelio Martín García Fernández,
		Richard Demo Souza	and 
		Hirley Alves
	}
	\thanks{O. López and H. Alves are with the Centre for Wireless Communications (CWC), University of Oulu, Finland. \{Onel.AlcarazLopez, hirley.alves\}@oulu.fi.}
	\thanks{E.M.G. Fernández is with Federal University of Paraná (UFPR), Curitiba, Brazil. evelio@ufpr.br.} 
	\thanks{R.D. Souza is with Federal University of Santa Catarina (UFSC), Florianópolis, Brazil. richard.demo@ufsc.br.}
	\thanks{This work was supported by CNPq, CAPES, Funda\c{c}\~ao Arauc\'aria (Brazil) and Academy of Finland (grant no 303532), and the Program for Graduate Students from Cooperation Agreements (PEC-PG, of CAPES/CNPq Brazil).}
}	
\maketitle

\begin{abstract}	
	We analyze a wireless communication system with finite block length and finite battery energy, under quasi-static Nakagami-m fading. Wireless energy transfer is carried out in the downlink while information transfer occurs in the uplink. Transmission strategies for scenarios with/without energy accumulation between transmission rounds are characterized in terms of error probability and energy consumption. A power control protocol for the energy accumulation scenario is proposed and results show the enormous impact on improving the system performance, in terms of error probability and energy consumption. The numerical results corroborate the existence and uniqueness of an optimum target error probability, while showing that a relatively small battery could be a limiting factor for some setups, specially when using the energy accumulation strategy. 
\end{abstract}


\section{Introduction}\label{Int}
The Internet of Things (IoT) is a recent communication paradigm which promises to bring wireless connectivity to ``...anything that may benefit from being connected...'' \cite{Dahlman.2014}, ranging from tiny static sensors to vehicles and drones. Consequently, coming wireless communication systems will have to support a much larger number of connected devices, including autonomous machines and devices, with applications having stringent requirements on latency and reliability as~\cite{Schulz.2017}: factory automation, with maximum latency around 0.25-10ms and maximum error probability of $10^{-9}$; smart grids (3-20ms, $10^{-6}$), professional audio (2ms, $10^{-6}$), etc. Powering and uninterrupted operation of such potential massive number of IoT nodes is a major challenge. \cite{Zanella.2014}. Energy harvesting (EH) techniques have recently drawn significant attention as a potential solution, and authors in \cite{Shaviv.2016} provide an insightful formula regardless of the type of energy source for the approximate capacity of the EH channel over Additive White Gaussian Noise (AWGN) for large and small battery regimes. Wireless Energy Transfer (WET) is a particularly attractive EH technique because radio-frequency (RF) signals can carry both energy and information, which enables energy constrained nodes to harvest energy and receive information \cite{Varshney.2008,Grover.2010}, allowing to prolong their lifetime almost indefinitely.

The exploitation of WET becomes very attractive specially for IoT scenarios where replacing or recharging batteries require high cost and/or can be inconvenient or hazardous (e.g., in toxic environments), or highly undesirable (e.g., for sensors embedded in building structures or inside the human body) \cite{Zhang.2013}. With further advances in antenna technology and EH circuit designs, WET is believed to become very efficient such that it will be implemented widely in the near future. Indeed, WET techniques are now evolving from theoretical concepts into practical devices for low-power electronic applications \cite{Powercast}. Wireless-powered communication networks (WPCNs), where the wireless terminals are powered only by WET and transmit their information using the harvested energy, have been widely investigated in the last years. The feasibility of WET for low-power cellular applications has been studied using experimental results, which have been summarized in \cite{Lu.2015}.  
A classic multi-user WPCN was investigated in \cite{Ju.2014}, where authors develop a ``harvest-then-transmit'' protocol which allows users to first collect energy from the signals broadcasted
by a single-antenna hybrid access-point (AP) in the downlink and then to use their harvested energy to send independent information to the hybrid AP in the uplink. Diverse strategies have been considered in the recent scientific literature in  order to improve the performance of WPCNs, such as relay-assisted \cite{Krikidis.2012,Gurakan.2012,Nasir.2013,Krikidis.2014,Ding.2014,Moritz.2014,Chen.2015,Nasir.2015,Michalopoulos.2015,Xiong.2015,Li.2016,Gu.2016,Mishra.2017}, Hybrid Automatic Repeat-reQuest (HARQ) \cite{Witt.2014}, and power control \cite{Huang.2013,Liu.2013,Isikman.2016,Gu.2016,Mishra.2017,Li.2016}, mechanisms. Works in \cite{Gu.2016,Mishra.2017,Li.2016} are particularly interesting since they propose energy accumulation strategies so that a wireless-powered relay can efficiently assist a communication link. Specifically, an accumulate-then-forward protocol for a multi-antenna relay is presented in \cite{Li.2016} while the charging/discharging behaviors of the relay battery are modeled as a finite-state Markov chain. The relay battery is modeled similarly in \cite{Gu.2016}, and the authors develop a power splitting-based energy accumulation scheme. Therein, a predefined energy threshold is set so the relay can determine whether it has sufficient energy to perform jointly energy accumulation and information forwarding. Otherwise, all the received signal power will be accumulated at the relay. Finally, a cooperative dilemma at the relay, concerning on whether to transfer its harvested energy to the source or to act as an information relay to the destination, is investigated in \cite{Mishra.2017}. Authors resolve this dilemma by providing insights into the optimal positioning suited for either energy relaying or information transfer.

All the above studies are under ideal assumption of communicating with large enough blocks in order to invoke Shannon theoretic arguments to address error performance. However, as pointed out in \cite{Makki.2016}, important characteristics of WET systems are: i) power consumption of the nodes on the order of $\mu$W; ii) strict requirements on the reliability of the energy supply and of the data transfer; iii) information is conveyed in short packets.
This third requirement is due to intrinsically small data payloads, low-latency requirements, and/or lack of energy resources to support longer transmissions \cite{Khan.2016}. This agrees well with several aforementioned IoT scenarios with stringent latency requirements. Although performance metrics like Shannon capacity, and its extension to nonergodic channels, have been proven useful to design current wireless systems, they are not necessarily appropriate in a short-packet scenario \cite{Durisi.2015}, where a more suitable metric is the maximum achievable rate at a given block length and error probability. This metric has been characterized in \cite{Polyanskiy.2010,Yang.2013} for both AWGN and fading channels. Indeed, recent works in finite-blocklength information theory have shed light on a number of cases where asymptotic results yield inaccurate engineering insights on the design of communication systems once a constraint on the codeword length is imposed, e.g., in fast fading scenarios and low-rate transmissions \cite{Polyanskiy.2011,Kostina.2013,Yang.2012,Durisi.2016,Mary.2015}. Recently, WPCNs under finite blocklength regime have received attention in the scientific community. In \cite{Lopez2.2017} we analyze and optimize a single-hop wireless system with energy transfer in the downlink and information transfer in the uplink, under quasi-static Nakagami-m fading in ultra-reliable communication (URC) scenarios, representative of wireless systems with strict error and latency requirements. The results demonstrate that there is an optimum number of channel uses for both energy and information transfer for a given message length. The impact of a decode-and-forward relay-assisted communication setup is evaluated in \cite{Lopez.2017} in terms of throughput and delay, also in URC scenario.
Achievable channel coding rate and mean delay of a point-to-point EH system with finite blocklength are investigated in \cite{Guo.2016} for an AWGN channel.
 On the other hand, subblock energy-constrained codes are investigated in \cite{Tandon.2016}, and a sufficient condition on the subblock length to avoid energy outage at the receiver is provided. In \cite{Khan.2016}, a node charged by a power beacon attempts to communicate with a receiver over a noisy channel. Authors investigate the impact of the number of channel uses for WET and for wireless information transfer (WIT) on the system performance. 
Also, tight approximations for the outage probability/throughput are given in \cite{Haghifam.2016} for an amplify-and-forward relaying scenario, while retransmission protocols, in both energy and information transmission phases, are implemented in \cite{Makki.2016} to reduce the outage probability compared to open-loop communication.

Moreover, power allocation strategies have been recently investigated to enhance the performance of short packets communication systems. In \cite{Lechner.2011}, the authors investigate the optimal power allocation algorithms for low-density parity-check (LDPC) codes with specific degree distributions using multi-edge-type density evolution error boundaries, while the error probability in delay-limited block-fading channels is analyzed.
A single point-to-point wireless link operating under queuing constraints, in the form of limitations on the buffer violation probabilities, is considered in \cite{Gursoy.2013}. The performance of different transmission strategies (e.g., variable-rate, variable-power, and fixed-rate transmissions) is also studied at finite blocklength regime. Furthermore, the maximum achievable channel coding rate at a given blocklength and error probability, when the codewords are subject to a long-term (e.g., averaged-over-all-codeword) power constraint is investigated in \cite{Yang.2015}, in which power control strategies for both AWGN and fading channels are developed. However, to the best of our knowledge, there are only few papers, e.g., \cite{Huang.2013,Liu.2013,Isikman.2016}, where power allocation strategies are proposed for WPCNs but based on the assumption of infinite blocklength. 
Particularly interesting is the work in \cite{Isikman.2016}, where authors propose a low-complexity solution, called fixed threshold transmission (FTT) scheme, and show that its performance is very close to the optimal. This strategy assumes a transmit power threshold to determine whether transmission takes place or not. If the channel state of the current transmission attempt is of poor quality, then saving energy for future transmission attempts may be a wiser choice. 

This paper aims at WPCN scenarios with short packets, but with several differences with respect to the related literature. The system is composed of a point-to-point communication link under Nakagami-m quasi-static fading, with WET in the downlink and WIT in the uplink, as in many of the related works. However, we analyze the error probability and average energy consumption under a finite battery constraint for scenarios with and without energy accumulation between transmission rounds while taking into account the sensitivity of the energy harvester, which is a parameter of practical interest. In addition, we propose a power control protocol for the scenario with energy accumulation between transmission rounds in order to enhance the system performance in terms of error probability by taking advantage of the channel state information (CSI) at the transmitter side, while at the same time the average energy consumption improves. The proposed strategy could be seen as a variant of the finite-blocklength scenarios of the FTT scheme investigated in \cite{Isikman.2016}. Notice that the infinite blocklength assumption in \cite{Isikman.2016} leads to a non-optimal transmit power threshold in our scenario, as the true required threshold is much higher when communicating with short packets.

The main contributions of this work can be listed as follows: 
\begin{itemize}
	\item Accurate closed-form approximations for the error probability in scenarios where all the energy harvested at each WET phase is used to transmit in the next WIT phase. Here, channels for WET and WIT phases are assumed reciprocal, which is different from the result in \cite{Lopez2.2017} where those channels are independent and only infinite battery setup is considered;
	\item A power control algorithm for scenarios with energy accumulation between transmission rounds. Proposed algorithm takes into account the message blocklength and consequently it can be seen as a more practical implementation of the FTT scheme proposed in \cite{Isikman.2016}. Notice that allowing energy accumulation between transmission rounds is beyond the scope of our previous work in \cite{Lopez2.2017};
	\item An analytical approach is provided that shows how misleading any scheme based on the assumption of infinite blocklength compared to finite blocklength is, which validates our assumptions and modeling;
	\item An analysis of the average energy consumption in addition to the error probability, which is not addressed in \cite{Isikman.2016} nor \cite{Lopez2.2017}, for scenarios with and without energy accumulation between transmission rounds. Saving energy for future transmissions allows to improve the system performance in terms of error probability while reducing the energy consumption. A relatively small battery could be a limiting factor for some setups, and specially when using the energy accumulation strategy which also depends heavily on the chosen target error probability.
	\end{itemize}

Next, Section \ref{system} presents the system model and assumptions. Section~\ref{HT} discusses a scenario without energy accumulation between transmission rounds, while the case with energy accumulation is analyzed in Section~\ref{PC} by proposing a power control protocol. Section~\ref{results} presents the numerical results. Finally, Section \ref{conclusions} concludes the paper.
\newline\textbf{Notation:} $X\sim\Gamma(m,1/m)$ is a normalized gamma distributed random variable with shape factor $m$, Probability Density Function (PDF) $f_X(x)=\frac{m^m}{\Gamma(m)}x^{m-1}e^{-mx}$ and Cumulative Distribution Function (CDF) $F_X(x)=1-\frac{\Gamma(m,mx)}{\Gamma(m)}$. Let $\mathds{E}[\!\ \cdot\ \!]$ denote expectation, $|\cdot|$ is the absolute value operator, and $\mathds{1}(\cdot)$ is an indicator function which is equal to $1$ if its argument is true and $0$ otherwise. Also, $\mathds{P}[A]$ is the probability of event $A$, while $\min(x,y)$ and $\max(x,y)$ are the minimum and maximum values between $x$ and $y$, respectively.

\section{System Model and Assumptions}\label{system}
Consider the point-to-point wireless communication system shown in Fig.~\ref{Fig1}, in which $S$ represents the information source,  $D$ is the destination, and both are single antenna, half-duplex, devices.
$D$ is assumed to be externally powered, while $S$ may be seen as a sensor node with very limited energy supply and finite battery. First, $D$ charges $S$ during $v$ channel uses in the WET phase, and doing that, acts as an interrogator, requesting information from $S$. Then, $S$ transmits $k$ information bits over $n$ channel uses in the WIT phase.
We define a ``transmission round'' as a pair of consecutive WET and WIT phases, in that order.
Notice that $S$ can transmit its data using all the energy available in its battery at the start of each WIT phase (without energy accumulation between transmission rounds) or just make use of a part of that energy, saving the rest for future transmissions (with energy accumulation between transmission rounds). We consider a time-constrained setup, which implies that $D$ has to decode the received signal for each arriving information block.

In addition, channel reciprocity holds as shown in Fig.~\ref{Fig1}, because we consider the same frequency bands for both WET and WIT phases\footnote{The reciprocity principle is based on the property that electromagnetic waves traveling in both directions will undergo the same physical perturbations. Therefore, if the link operates on the same frequency band in both directions, the impulse response of the channel observed between any two antennas should be the same regardless of the direction \cite{Guillaud.2005}. This is a very common assumption in many works related with WPCNs, e.g., \cite{Witt.2014,Zhao.2016,Hadzi.2016,Yang.2015.2}. In practice, the non-symmetric characteristics of the RF electronic circuitry would affect the reciprocity property, and some calibration methods would be required \cite{Guillaud.2005}.}. We assume low-mobility scenarios, for which the coherence time is large enough such that channels are quasi-static, e.g., the fading process is constant over a transmission round ($v+n$ channel uses) and independent and identically distributed from round to round. To support this assumption, let $T_{h}$ be the coherence time, thus $T_h\approx\frac{1}{f_m}=\frac{c}{fv_d}$, 
where $f_m=\tfrac{v_d}{c}f$ is the maximum Doppler spread, $v_d$ is the device velocity, $c$ is the speed of light and $f$ is the transmission frequency. 
Fig.~\ref{Fig2}a shows the approximate coherence time as a function of the device velocities for several transmission frequencies. For velocities below $20$km/h and for all the frequencies being considered, the coherence time is expected to be above $10$ms. Also, if we fixed the transmission frequency to $2$GHz while selecting low-mobility scenarios, we can see in Fig.~\ref{Fig2}b the coherence time in channel uses as a function of the channel use duration. For low velocities ($v_d\le 3$km/h), the coherence time is always above $1000$ channel uses. 
Therefore, for low mobility scenarios and finite (short) blocklength, the quasi-static assumption holds, as well as the reciprocity of the channels, since it is expected that $v+n\le \frac{T_h}{T_c}$.

The fading is modeled using the Nakagami-m distribution, which is a generalized distribution that can model different fading environments by adjusting its parameters to fit a variety of empirical measurements. In fact, multipath fading can be adequately characterized by the Nakagami-m distribution, and it can model also the Rayleigh and Rician distributions, as well as more general ones \cite{Goldsmith.2005}. 
We consider normalized channel gains, then $g=|h|^2\sim\Gamma(m,1/m)$, while the duration of a channel use is denoted by $T_c$.

In the scenario without energy accumulation between transmission rounds, perfect CSI is assumed only at $D$ when decoding after the WIT phase. For the scenario with energy saving, CSI is also assumed at the transmitter side. Although CSI acquisition in an energy-limited setup is not trivial, our analysis based on perfect CSI gives an upper-bound on the performance of real scenarios, where additional delay and imperfections in channel estimation are present. In addition, notice that CSI at the transmitter side can be acquired via feedback from $D$ or even if $D$ sends pilots taking advantage of channel reciprocity. 
\begin{figure}[t!]
	\centering
	\subfigure{\includegraphics[width=0.4\textwidth]{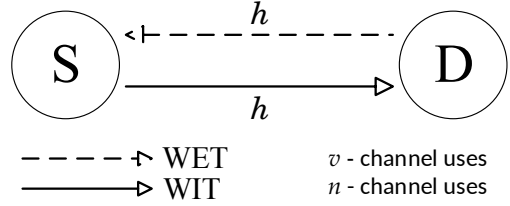}}
	\vspace*{-3mm}
	\caption{System model with WET in the downlink and WIT in the uplink.}		
	\label{Fig1}
	\vspace*{-5mm}
\end{figure}
\begin{figure}[h!]
	\centering
	\includegraphics[width=0.65\textwidth]{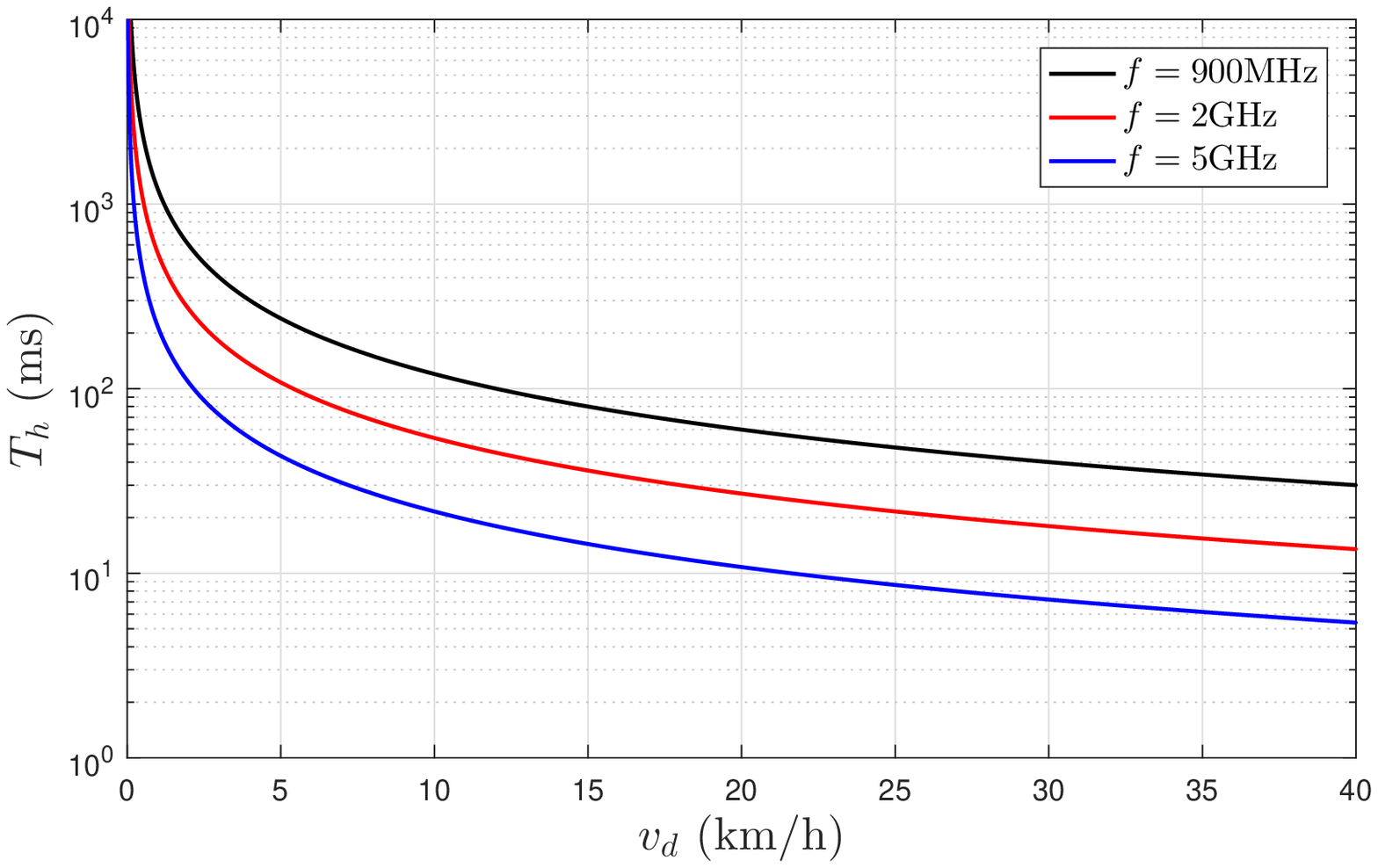}\\
	\vspace{2mm}
	\includegraphics[width=0.65\textwidth]{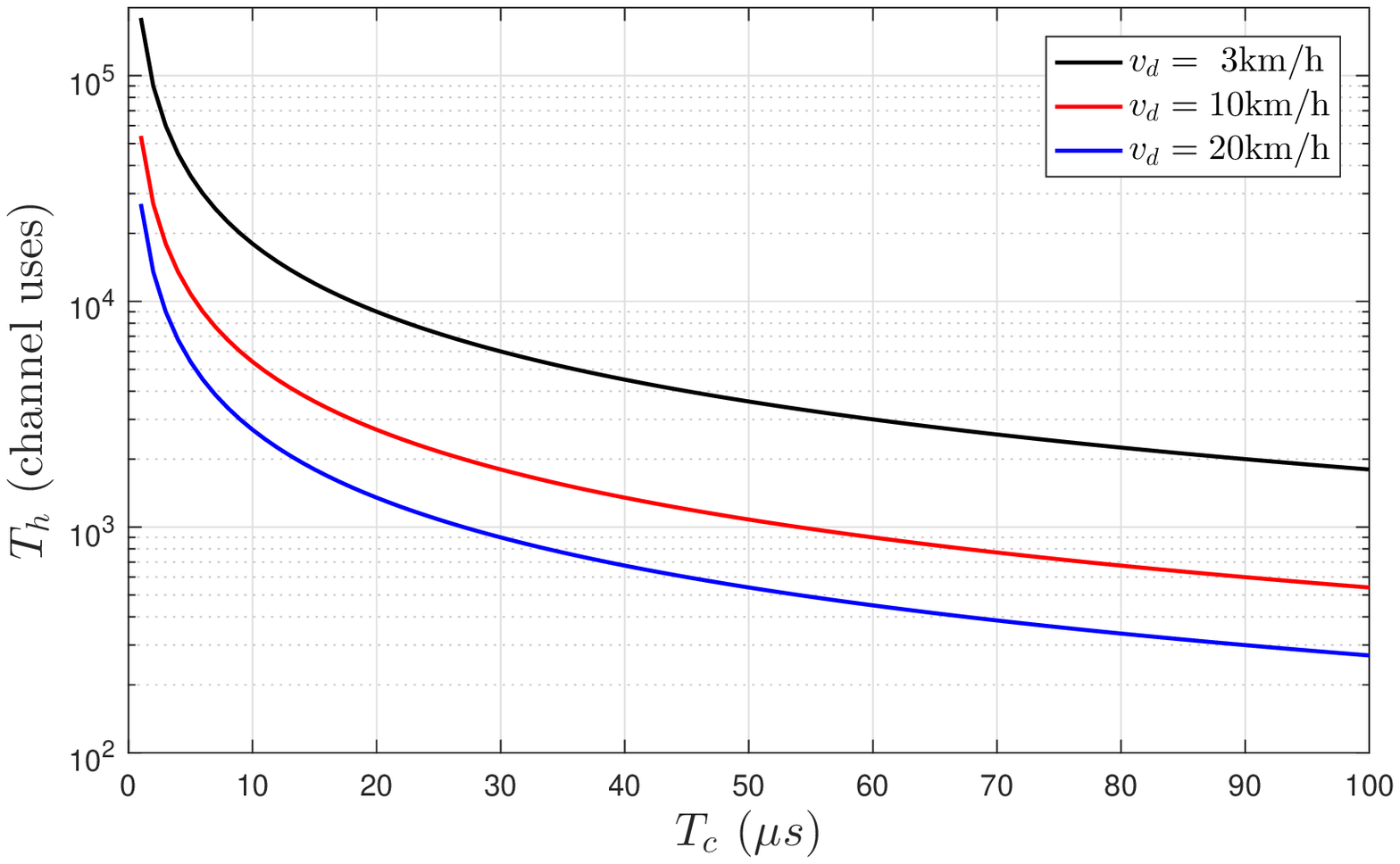}
	\vspace{-3mm}
	\caption{Coherence time a) in ms and as a function of $v_d$ for $f\in\{0.9,2,5\}$GHz (top), and b) in channel uses, and as a function of the duration of a channel use, $T_c$, for $f=2$GHz and $v_d\in\{3,10,20\}$km/h (bottom).}\label{Fig2}		
	\vspace*{-5mm}
\end{figure}
\section{Harvest then Transmit (HTT)}\label{HT}
In this section we analyze the scenario without energy accumulation between transmission rounds.
\subsection{WET Phase}
In this phase, $D$ charges $S$ during $v$ channel uses. The receiving power at $S$ is
\begin{equation}\label{Pr}
P_{r,i}=\frac{P_d |h_i|^2}{\kappa d^{\alpha}}=\frac{P_d g_i}{\kappa d^{\alpha}},
\end{equation}
\noindent where $P_d$ is the transmit power of $D$, $d$ is the distance between $S$ and $D$, $\alpha$ is the path loss exponent and $\kappa$ accounts for other factors as the carrier frequency, heights and gains of the antennas \cite{Goldsmith.2005}. Now, the energy harvested at $S$ during the $i$th transmission round is
\begin{equation}\label{Eh}
E_i=\mathds{1}(P_{r,i}\ge \varpi)\eta P_{r,i}vT_c=\mathds{1}(g_i\ge \varpi^*)\eta P_{r,i}vT_c,
\end{equation}
\noindent where $\varpi$ is the sensitivity of the energy harvester (minimum RF input power required for energy harvesting), therefore $\varpi^*=\frac{\varpi\kappa d^{\alpha}}{P_d}$ is the channel sensitivity threshold, while $0<\eta<1$ is the energy conversion efficiency. The indicator function allows to make $E_i=0$ for any received power below of the sensitivity level. In addition, we assume that $P_d$ is sufficiently large such that the energy harvested from noise is negligible.
Harvested energy is first stored in a rechargeable battery of capacity $B_{_{\mathrm{max}}}$, and becomes available in the current round. Then, the charge of the battery at the beginning of the $i$th WIT phase is updated as follows
\begin{align}
B_i&=\min(B_{\mathrm{max}},E_i)=\min\Big(B_{\mathrm{max}},\mathds{1}(g_i\ge \varpi^*)\frac{\eta P_d}{\kappa d^{\alpha}}vT_c\Big)=\mathds{1}(g_i\ge \varpi^*)\min\Big(B_{\mathrm{max}},\frac{\eta P_d}{\kappa d^{\alpha}}vT_c\Big)\nonumber\\
&=\mathds{1}(g_i\ge \varpi^*)\min(g_i,\lambda)\frac{\eta P_d}{\kappa d^{\alpha}}vT_c,
\end{align}
\noindent where $\lambda=\frac{B_{\mathrm{max}}\kappa d^{\alpha}}{\eta P_d vT_c}$ is the channel power gain threshold for the saturation of the battery in $S$.
\subsection{WIT Phase}
After energy has been harvested during the WET phase, which is only when received power overcomes the sensitivity of the energy harvester, $S$ uses all the energy in its battery to transmit a message of $k$ bits to $D$ over $n$ channel uses. The signal received at $D$ during the $i$th round can be written as
\begin{equation}\label{yd}
\mathbf{y}_{d,i}=\sqrt{\frac{P_{s,i}}{\kappa d^{\alpha}}}h_i\mathbf{x}_{s,i}+\mathbf{w}_{d,i},
\end{equation}
where $\mathbf{x}_s$ belongs to the zero-mean, unit-variance Gaussian codebook transmitted by $S$, $\mathds{E}[|\mathbf{x}_s|^2]=1$,  $\mathbf{w}_d$ is the Gaussian noise vector at $D$ with variance $\sigma_d^2$ and 
\begin{equation}\label{Ps}
P_{s,i}=\frac{B_i}{nT_c}=\frac{\eta vP_d}{n\kappa d^{\alpha}}\min(g_i,\lambda)\mathds{1}(g_i\ge \varpi^*)
\end{equation}
\noindent is the transmit power.
Thus, the instantaneous Signal-to-Noise Ratio (SNR) at $D$ in the $i$th round is
\begin{equation}\label{SNR}
\gamma_i=\frac{P_{s,i}g_i}{\kappa d^{\alpha}\sigma_d^2}=\frac{\eta vP_dg_i}{n\kappa^2d^{2\alpha}\sigma_d^2}\min(g_i,\lambda)\mathds{1}(g_i\ge \varpi^*)=\beta g_i\min(g_i,\lambda)\mathds{1}(g_i\ge \varpi^*),
\end{equation}
\noindent which is proportional to the square of the power channel coefficient as long as the $S$ battery is not saturated, otherwise it only relies on the scaled power channel coefficient, and $\beta=\frac{\eta vP_d}{n\kappa^2d^{2\alpha}\sigma_d^2}$.
\subsection{Error Probability and Average Power Consumption}
The information theoretic analysis for infinite blocklength says that no error occurs as long as $\gamma>2^r-1$ \cite{Goldsmith.2005}. However, if we communicate over a noisy channel and we are restricted to use a finite number of channel uses, then no protocol is able to achieve perfectly reliable communication \cite{Bertsekas.1992}.
Let $\epsilon_i$ be the error probability for the information block transmitted in the $i$th round, which is well approximated by \cite[Eq.(5)]{Hu.2016} 
\begin{equation}\label{e1}
\epsilon_i\approx Q\Biggl(\frac{C(\gamma_i)-r}{\sqrt{V(\gamma_i)/n}}\Biggl),
\end{equation}
where $r = k/n$ is the source fixed transmission rate,  $C(\gamma_i)=\log_2(1+\gamma_i)$ is the Shannon capacity, $V(\gamma_i)=\left(1-\frac{1}{(1+\gamma_i)^2}\right)(\log_2e)^2$ is the channel dispersion, which measures the stochastic variability of the channel relative to a deterministic channel with the same capacity \cite{Polyanskiy.2010}, and $Q(x)=\int_{x}^{\infty}\frac{1}{\sqrt{2\pi}}e^{-t^2/2}\mathrm{d}t$. 
For quasi-static fading channels the error probability is \cite[eq.(59)]{Yang.2013}
\begin{align}
\varepsilon&=\mathds{E}[\epsilon_i]\approx\int\limits_{0}^{\infty}\!Q\Biggl(\!\frac{C(\gamma_i)\!-\!r}{\sqrt{V(\gamma_i)/n}}\!\Biggl)\!f_G(g)\mathrm{d}g\nonumber\\
&\!\stackrel{(a)}{=}\underbrace{F_G(\varpi^*)}_{\varepsilon_1}+\!\underbrace{\int\limits_{\min(\varpi^*,\lambda)}^{\lambda}\!Q\Biggl(\!\frac{C(\beta g^2)\!-\!r}{\sqrt{\frac{V(\beta g^2)}{n}}}\!\Biggl)\!f_G(g)\mathrm{d}g\!+\!\int\limits_{\max(\varpi^*,\lambda)}^{\infty}\!Q\Biggl(\!\frac{C(\beta\lambda g)\!-\!r}{\sqrt{\frac{V(\beta\lambda g)}{n}}}\!\Biggl)\!f_G(g)\mathrm{d}g}_{\varepsilon_2}, \label{EX}
\end{align}
where $(a)$ comes from using \eqref{SNR}. Notice that $\varepsilon_1=F_G(\varpi^*)=1-\tfrac{\Gamma(m,m\varpi^*)}{\Gamma(m)}$, accounts for situations where the power transfer is unsuccessful because of the sensitivity of the energy harvester, while $\varepsilon_2$ is the error probability when communicating. Both, \eqref{e1} and $\varepsilon_2$ in \eqref{EX}, are accurate when considering blocklength $n\ge 100$ as shown in \cite[Figs. 12 and 13]{Polyanskiy.2010} for AWGN, and in \cite{Yang.2014_2} for fading channels, respectively. Notice that it seems intractable to find a closed-form solution for $\varepsilon_2$ in \eqref{EX}. Then, first we resort to the approximation of $Q\big(p(\mu g^t)\big)$, $p(\mu g^t)=\frac{C(\mu g^t)-r}{\sqrt{V(\mu g^t)/n}}$, given by \cite{Makki.2016,Makki.2014}
\begin{align}\label{AP}
Q(p(\mu g^t))\!\approx\!\Omega(\mu g^t)\!=\!\left\{\begin{array}{ll}
\!1,&  g\le \zeta^{\tfrac{2}{t}}\\
\!\frac{1}{2}\!-\!\frac{\phi}{\sqrt{2\pi}}(\mu g^t\!-\!\theta)\!,\!&\!\zeta^{\tfrac{2}{t}}\!<\!g\!<\!\varphi^{\tfrac{2}{t}}\!\\
\!0,& g\ge \varphi^{\tfrac{2}{t}}
\end{array}
\right.\!,
\end{align}
where $\zeta=\sqrt{\tfrac{\varrho}{\mu}}$,  $\varphi =\sqrt{\tfrac{\vartheta}{\mu}}$,  $\theta=2^{r}-1$, $\phi=\sqrt{\frac{n}{2\pi}}(2^{2r}-1)^{-\frac{1}{2}}$, $\varrho=\theta-\frac{1}{\phi}\sqrt{\frac{\pi}{2}}$ and $\vartheta=\theta+\frac{1}{\phi}\sqrt{\frac{\pi}{2}}$, which leads to the following result.
\begin{theorem}\label{prop_1}
	For the system described in Section \ref{HT}, the error probability when communicating, $\varepsilon_2$ in \eqref{EX}, can be approximated as in \eqref{AP2} and \eqref{AP3} for finite and infinite battery devices, respectively, where $\omega_1=\big(\frac{1}{2}+\frac{\phi\theta}{\sqrt{2\pi}}\big)$, $\omega_2=\frac{\phi\beta}{\sqrt{2\pi}}$, $z_{11}=\min(\zeta_1,\lambda)$, $z_{12}=\min(\varphi_1,\lambda)$, $z_{13}=\min(z_{11},\varpi^*)$, $z_{14}=\min\big(\max(z_{11},\varpi^*),z_{12}\big)$, $z_{15}=\min(\zeta_1,\varpi^*)$, $z_{16}=\min(\max(\zeta_1,\varpi^*),\varphi_1)$, $z_{21}=\max(\zeta_2^2,z_{23})$, $z_{22}=\max(\varphi_2^2,z_{23})$, $z_{23}=\max(\lambda,\varpi^*)$, and $\zeta_j=\sqrt{\tfrac{\varrho}{\mu_j}}$, $\varphi_j=\sqrt{\tfrac{\vartheta}{\mu_j}}$, $\mu_j=\beta\lambda^{j-1}$, with $j\in\{1,2\}$.
	\begin{figure*}[!h]	
		\vspace{-6mm}
		\small
		\begin{align} 
		&\varepsilon_2\!\approx\! \frac{1}{\Gamma(m)}\bigg[\Gamma\big(m,mz_{13}\big)\!+\!\Gamma\big(m,mz_{23}\big)\!-\!\Gamma\big(m,mz_{11}\big)\!+\!\omega_1\Big(\Gamma\big(m,mz_{14}\big)\!-\!\Gamma\big(m,mz_{12}\big)\!-\!\Gamma\big(m,mz_{22}\big)\Big)\!+\!\nonumber\\
		&\!+\!\frac{\omega_2}{m^2}\!\Big(\!\Gamma\!\big(m\!+\!2,mz_{12}\big)\!-\!\Gamma\!\big(m\!+\!2,mz_{14}\big)\!\Big)\!+\!(\omega_1\!-\!1)\Gamma\big(m,mz_{21}\big)\!+\!\frac{\omega_2z_{23}}{m}\!\Big(\!\Gamma\!\big(m\!+\!1,mz_{22}\big)\!-\!\Gamma\!\big(m\!+\!1,mz_{21}\big)\!\Big)\!\bigg]\label{AP2}
		\end{align}
		\begin{align} 
		\varepsilon_{_{2,\infty}}&\!\!\approx\!\! \frac{1}{\Gamma(m)}\!\!\bigg[\!\Gamma\big(\!m,mz_{15}\!\big)\!-\!\Gamma\big(\!m,m\zeta_1\!\big)\!+\!\omega_1\!\Big(\!\Gamma\big(\!m,mz_{16}\!\big)\!-\!\Gamma\big(\!m,m\varphi_1\!\big)\!\Big)\!+\!\frac{\omega_2}{m^2}\!\Big(\!\Gamma\big(\!m\!+\!2,m\varphi_1\!\big)\!-\!\Gamma\big(\!m\!+\!2,mz_{16}\!\big)\!\Big)\!\bigg]\label{AP3}
		\end{align}
	\vspace{-6mm}
	\end{figure*}
\end{theorem}
\begin{proof}
	See Appendix~\ref{App_A}.
		\phantom\qedhere
\end{proof}
\begin{remark}
	Differently from \cite[eq.(10)]{Lopez2.2017}, where battery is assumed infinite and the WET and WIT channels are considered independent, results in Theorem~\ref{prop_1} hold for both, finite and infinite battery, and considering reciprocal channels and the sensitivity of the energy harvester at the receiver.
\end{remark}

In order to mathematically characterize the energy consumption under the HTT protocol operation, we state the following theorem.
\begin{theorem}\label{prop_2}
	The average transmit power of node $S$ when using the Harvest then Transmit protocol is given by
	\begin{align}
	\bar{P}&=\frac{\eta vP_d}{n\kappa d^{\alpha}}\bigg[\frac{\Gamma(m+1,m\varpi^*)-\Gamma(m+1,m\tau)}{\Gamma(m+1)}+\lambda\frac{\Gamma(m,m\tau)}{\Gamma(m)}\bigg],\label{Pfin}\\
	\bar{P}_{_{\infty}}&=\frac{\Gamma(m+1,m\varpi^*)}{\Gamma(m+1)}\frac{\eta vP_d}{n\kappa d^{\alpha}},\label{Pinf}
	\end{align}
	for finite and infinite battery devices, respectively, and $\tau=\max(\varpi^*,\lambda)$.
\end{theorem}	
\begin{proof}
See Appendix~\ref{App_B}.
\phantom\qedhere
\end{proof}

An interesting fact from \eqref{Pfin} with $\varpi^*=0$ is that for scenarios in which the fading is less severe, e.g., larger $m$, the average power consumption increases asymptotically approaching the case of infinite battery \eqref{Pinf} for practical systems where $\lambda>1$. This is
	\begin{align}
	\lim\limits_{m\rightarrow\infty}\bar{P}=\frac{\eta vP_d}{n\kappa d^{\alpha}}\min{(\lambda,1)},
	\end{align}
which makes sense since the channel tends to behave like an AWGN channel and no battery saturation occurs for $\lambda>1$.

Notice that the average energy consumption can be computed as $nT_c\bar{P}$ or $nT_c\bar{P}_{_{\infty}}$ for finite and infinite battery devices, respectively.
\section{Allowing Energy Accumulation}\label{PC}
Herein we analyze a scenario with energy accumulation between transmission rounds, where the charge of the battery is now given by 
\begin{align}\label{b1}
B_i=\min(B_{\mathrm{max}},B_{i-1}+E_i-P_{s,i-1}nT_c),
\end{align}
with $E_i$ obeying \eqref{Eh}. We develop  a power control protocol, with channel knowledge at $S$, in order to improve the performance of the  scenario discussed in Section~\ref{HT}. It is important to note that the model is independent from the channel fading distribution, not being restrict to Nakagami-m fading.
\subsection{Power Control Strategy}
Authors in \cite{Isikman.2016} propose the FTT power control protocol, based on the assumption that transmitting with certain power such that $\gamma=2^r-1$ is  sufficient for error-free communication. However, this assumption is not true, not even when $\gamma\gg2^r-1$ in a finite blocklength setup. How accurate are the results of the FTT scheme in a finite blocklength scenario is a question answered later in Subsection~\ref{seccC}, but first let us propose a finite blocklength variant of the FTT protocol (FB-FTT), which can be summarized as follows.
\begin{enumerate}
	\item Let $\epsilon_{_{\mathrm{th}}}$ be the maximum error probability at $D$.
	\item Then, $S$ chooses a transmit power $\hat{P}_s$, so that $\epsilon_i=\epsilon_{_{\mathrm{th}}}$ according to \eqref{e1}.
	\item If there is not enough energy in the battery for $S$ to transmit with sufficient power, then $S$ stays silent and saves energy for the next transmission round. That WIT block is considered lost, and $S$ attempts to transmit a new WIT block at the next round.
\end{enumerate}
Notice that the idea behind the FB-FTT strategy is to transmit with a power that allows to achieve a given SNR $\hat{\gamma}$ at $D$, which causes an error probability $\epsilon_{_{\mathrm{th}}}$, as long as there is a transmission from $S$. Then, knowing $\hat{\gamma}$ and the noise power at $D$, and based on \eqref{SNR}, the required transmit power is
\begin{align}\label{Psi}
\hat{P}_{s,i}=\frac{\hat{\gamma} \kappa d^{\alpha}\sigma_d^2}{g_i}.
\end{align}
Therefore, \eqref{b1} can be rewritten as
\begin{align}\label{b2}
B_i=\min\Big(&B_{_{\mathrm{max}}},E_i+B_{i-1}-\hat{P}_{s,i-1}nT_c\mathds{1}{\big(B_{i-1}>\hat{P}_{s,i-1}nT_c\big)}\!\Big),
\end{align}
where the value of the indicator function is $1$ if there is sufficient energy to support the transmission or $0$ otherwise, regulating the energy expenditure and, therefore, the state of charge at each round. Also, finding in closed form the required $\hat{\gamma}$ to reach an $\epsilon_{_{\mathrm{th}}}$ is algebraically impossible since we would have to solve for $\gamma$ the following approximate equation coming from \eqref{e1} \cite{Polyanskiy.2010}
\begin{align}\label{eq}
r\approx\log_2(1+\gamma)-\sqrt{1-\frac{1}{(1+\gamma)^2}}\frac{\log_2e\ Q^{-1}(\epsilon_{_{\mathrm{th}}})}{\sqrt{n}}.
\end{align}
\noindent However, since the required $\hat{\gamma}$ is fixed for each system setup $(n,k,\epsilon_{_{\mathrm{th}}})$, $S$ does not need to compute $\hat{\gamma}$  often\footnote{Notice that the value of the required $\hat{\gamma}$ could be even programmed in $S$ from the very beginning for static scenarios.}, and an iterative method is proposed in Algorithm~\ref{alg_1} to solve \eqref{eq}.
The idea is to iterate over
\begin{align}\label{gam}
\hat{\gamma}^{(t)}=2^{r+\frac{1}{\sqrt{n}}M^{(t-1)}\log_2e\ Q^{-1}(\epsilon_{_{\mathrm{th}}})}-1,
\end{align}
which comes from isolating $\gamma$ in \eqref{eq}, while abandoning the approximation notation by the equality, and using 
 \begin{align}\label{M}
 M^{(t)}=\sqrt{1-\frac{1}{(1+\hat{\gamma}^{(t)})^2}},
 \end{align}
 where $t$ is the iteration index. The choice for $M^{(0)}=1$ comes from the fact that this is a good approximation for high SNR. Also, $\gamma_{_{\Delta}}$ is the acceptable maximum difference between the value of $\hat{\gamma}$ found by the proposed algorithm and its real required value.
\begin{lemma}\label{conv}
	The required $\hat{\gamma}$ to reach $\epsilon_{_{\mathrm{th}}}$ is a unique solution to \eqref{gam}, and Algorithm~\ref{alg_1} (on the top of the next page) converges for $\epsilon_{_{\mathrm{eth}}}\le 0.5$.
\end{lemma}
\begin{proof}
	See Appendix~\ref{App_C}.
	\phantom\qedhere
\end{proof}
According to Lemma~\ref{conv}: $\hat{\gamma}=\hat{\gamma}^{(\infty)}$ if $\gamma_{_{\Delta}}\rightarrow 0$, and we denote $M=M^{(\infty)}=\sqrt{V(\hat{\gamma})}\ln 2$. The required number of iterations for a given precision $\gamma_{_{\Delta}}$, is numerically investigated in Section~\ref{SecConv}.

\begin{algorithm}[t!]
	\caption{Finding the required $\hat{\gamma}$ for a given $(n,k,\epsilon_{_{\mathrm{th}}})$}
	\label{alg_1}		
	\begin{algorithmic} [1]
		\State $t=1,\ M^{(0)}=1$  \label{line01} 
		\State Calculate $\hat{\gamma}^{(t)}$ using \eqref{gam}\label{line02} 
		\State Calculate $M^{(t)}$ using \eqref{M}\label{line03}                    
		\If{$|\hat{\gamma}^{(t)}-\hat{\gamma}^{(t-1)}|>\gamma_{_{\Delta}}$} \label{line04}    
		\State $t\rightarrow t+1$ \label{line05} 
		\State Return to line \ref{line02}						\label{line06}
		\EndIf        \label{line07}                     
		\State End \label{line08}         
	\end{algorithmic}
\end{algorithm}
\begin{remark}
	When achieving $\epsilon_i=\epsilon_{_{\mathrm{th}}}$ is impossible, an alternative strategy could be transmitting with the maximum available power that the harvested energy allows. We refer to this second strategy as Finite Blocklength Fixed Threshold Uninterrupted Transmission (FB-FTUT) and it is only included in Section~\ref{results} in order to assess the performance of the FB-FTT protocol.
\end{remark}

\subsection{Overall Error Probability and Mean Power Consumption}\label{lemma1}
The overall error probability for this scenario is given by
\begin{align}\label{outage}
\varepsilon=(1-\epsilon_{_{\mathrm{out}}})\epsilon_{_{\mathrm{th}}}+\epsilon_{_{\mathrm{out}}},
\end{align}
where $\epsilon_{_{\mathrm{out}}}=\mathds{P}[B_i<\hat{P}_{s,i}nT_c]$ is the probability that the  energy available in the battery is insufficient to achieve the required $\gamma$ at $D$ for a given target error $\epsilon_{_{\mathrm{th}}}$. 

Notice that $\epsilon_{_{\mathrm{out}}}$ depends on the value of $\epsilon_{_{\mathrm{th}}}$. The higher the value of $\epsilon_{_{\mathrm{th}}}$, smaller $\hat{\gamma}$ and transmit power $\hat{P}_s$ are required, and the smaller the value of $\epsilon_{_{\mathrm{out}}}$, and vice versa. Unfortunately, it seems intractable to find a closed-form expression for $\epsilon_{_{\mathrm{out}}}$ due to the complexity of \eqref{b2}, and we resort to simulations in Section~\ref{results} in order to compute it. In addition, we can notice that $\varepsilon\ge\epsilon_{_{\mathrm{th}}}$ always, thus a relatively high value of $\epsilon_{_{\mathrm{th}}}$ can seriously limit the system performance for some setups. 
%
In fact,
	numerical evidence suggests that there is a unique optimum value of $\epsilon_{_{\mathrm{th}}}$, $\epsilon_{_{\mathrm{th}}}^*$, that minimizes the overall error probability in practical setups 
	(see Appendix~\ref{App_D}).

In addition, the energy consumption of $S$, characterized in terms of its average transmit power, is as follows
\begin{theorem}\label{prop_3}
	The average transmit power of $S$ when using the proposed power control scheme, and when $m>1$, is 
	\begin{align}
	\bar{P}&=(1\!-\!\epsilon_{_{\mathrm{out}}})\gamma\kappa d^{\alpha}\sigma_d^2m\bigg[\frac{1}{m\!-\!1}\!-\!\frac{\Gamma(m\!-\!1,m\lambda)}{\Gamma(m)}\bigg],\label{PC1}\\
	\bar{P}_{_{\infty}}&=(1\!-\!\epsilon_{_{\mathrm{out}}})\gamma\kappa d^{\alpha}\sigma_d^2\frac{m}{m\!-\!1},\label{PC2}
	\end{align}
	for finite and infinite battery devices, respectively.
\end{theorem}	
\begin{proof}
	See Appendix~\ref{App_E}.
	\phantom\qedhere
\end{proof}

Notice that when the fading is less severe, e.g., larger $m$, $\epsilon_{_{\mathrm{out}}}$ decreases while the remaining terms depending on $m$ tend to unity. That is because asymptotically, and considering $\varpi^*<1$, which has to be true in practice, we have that
\begin{align}
\lim\limits_{m\rightarrow\infty}\bar{P}&=\gamma\kappa d^{\alpha}\sigma_d^2\mathds{1}(\lambda>1),\\
\lim\limits_{m\rightarrow\infty}\bar{P}_{\infty}&=\gamma\kappa d^{\alpha}\sigma_d^2.
\end{align}
Therefore, it is expected an average transmit power very close to $\gamma\kappa d^{\alpha}\sigma_d^2$ for any practical system with $m\gg1$. In fact, the instantaneous transmit power tends to be exactly $\gamma\kappa d^{\alpha}\sigma_d^2$ since the channel tends to an AWGN channel, at the same time that no saturation or complete depletion of the battery ever occurs.

\subsection{How accurate is the FTT power control protocol at finite blocklength?\label{seccC}}
We know that an error-free communication setup is unreachable for any practical system at finite blocklength. Thus, the FTT strategy presented in \cite{Isikman.2016} is over optimistic in a finite blocklength scenario as it relies on the fact that only $\gamma=2^r-1$ is required for full transmit reliability. Using a power that allows reaching an SNR equal or very close to $2^r-1$, in order to  save energy while  increasing the chances of future transmissions, would lead to error probabilities close to $0.5$ for short blocklengths. Even when we go further than the limit of $2^r-1$, there are still certain chances of error while at the same time the chances of future transmissions are decreased since the energy saving process is negatively affected.

Algorithm~\ref{alg_1} aims at finding the required SNR, $\hat{\gamma}$, for certain required reliability. Of course, this value would be greater than $2^r-1$ for any practical setup, e.g., $\epsilon_{_{\mathrm{th}}}<0.5$. However, an interesting question is \textit{how greater the $\hat{\gamma}$ would be when compared with the limit $2^r-1$, which is only valid at infinite blocklength?} In order to shed some light on that matter we define $\delta$ as the quotient between the required SNR considering a finite blocklength and the asymptotic SNR limit for error-free communication at infinite blocklength, thus
\begin{align}\label{delta}
\delta&=\frac{\hat{\gamma}}{2^r-1}\stackrel{(a)}{=}\frac{2^{r+\frac{1}{\sqrt{n}}M\log_2e\ Q^{-1}(\epsilon_{_{\mathrm{th}}})}-1}{2^r-1}\stackrel{(b)}{=}e^{\frac{1}{\sqrt{n}}M Q^{-1}(\epsilon_{_{\mathrm{th}}})}+\frac{e^{\frac{1}{\sqrt{n}}M Q^{-1}(\epsilon_{_{\mathrm{th}}})}-1}{2^r-1},
\end{align}
where $(a)$ comes from using \eqref{gam} with $t\rightarrow\infty$, although notice that $M$ is still a function of $\hat{\gamma}$, and $(b)$ comes from algebraic transformations. When high data rates are required, e.g., high SNR regime, we have that
\begin{align}\label{lim}
\lim\limits_{r\rightarrow\infty}\delta=e^{\frac{1}{\sqrt{n}} Q^{-1}(\epsilon_{_{\mathrm{th}}})},
\end{align}
since $M\rightarrow 1$ when $\hat{\gamma}\rightarrow\infty$, which is a lower bound on $\delta$ because $\delta$ is a decreasing function of $r$. 
\begin{figure}[t!]
	\centering
	\subfigure{\includegraphics[width=0.65\textwidth]{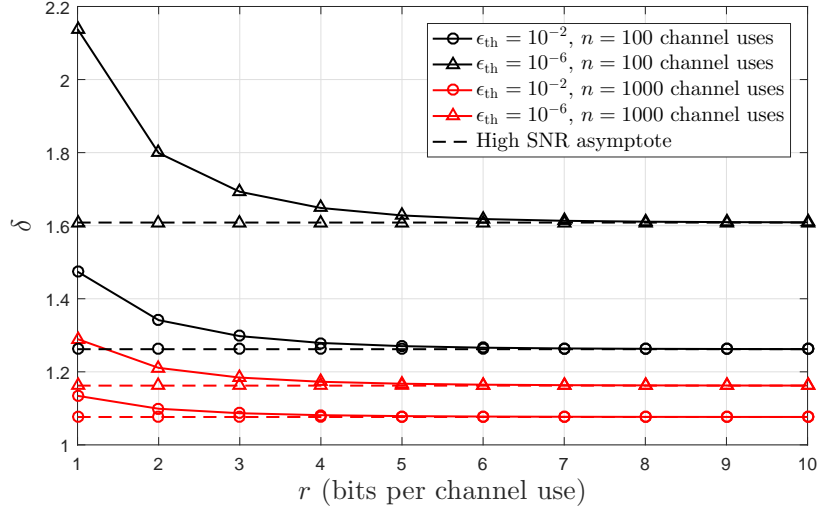}}
	\vspace{-3mm}
	\caption{$\delta$ as a function of $r$ for $n\in\{100,1000\}$ channel uses and $\epsilon_{_{\mathrm{th}}}\in\{10^{-2},\ 10^{-6}\}$.}	\label{Fig3}
	\vspace*{-5mm}	
\end{figure}
Fig.~\ref{Fig3} illustrates this behavior for $n\in\{100,1000\}$ channel uses and $\epsilon_{_{\mathrm{th}}}\in\{10^{-2},\ 10^{-6}\}$. The main remark from \eqref{lim}, which is also shown in the figure, is that the required SNR has to be at least $e^{\frac{1}{\sqrt{n}} Q^{-1}(\epsilon_{_{\mathrm{th}}})}$ times greater than the usual threshold of $2^r-1$ to reach an error probability no lower than $\epsilon_{_{\mathrm{th}}}$ while transmitting the information through $n$ channel uses. Obviously, this criterion is also applied to the transmit power. In fact, $\delta$ also approximates well to the quotient between the mean consumption power of the FB-FTT protocol and the FTT \cite{Isikman.2016} for practical scenarios where $\epsilon_{_{\mathrm{out}}}\ll1$, since $\frac{1-\epsilon_{_{\mathrm{out}}}^{\mathrm{FB-FTT}}}{1-\epsilon_{_{\mathrm{out}}}^{\mathrm{FTT}}}\approx 1$. Back to Fig.~\ref{Fig3}, notice that the asymptotic bound begins to be very tight already for $r\sim 4$ since $M>\sqrt{1-\tfrac{1}{(2^4)^2}}=0.998$, and $2^4-1=15$ is at least 10 times greater than the numerator of the fraction in the last equality in \eqref{delta} for any combination of $n\ge100$ channel uses and $\epsilon_{_{\mathrm{th}}}\ge 3\times 10^{-20}$. However, the lower the data rate and/or the shorter the information blocklength and/or the more stringent the target error probability, the greater the required SNR with respect to the threshold at infinite blocklength, thus showing how misleading is to calculate the SNR and rates using any scheme based on the assumption of infinite blocklength, such as the FTT protocol.

\section{Numerical Results}\label{results}
In this section, we present numerical results to investigate the performance of the proposed scheme as a function of the system parameters. Unless stated otherwise, results are obtained by setting, $P_d=3$W, $\alpha=3$, $d=9.8$m, and $\kappa=20$ dB is the average signal power attenuation at a reference distance of 1 meter. These values were chosen to provide an average power of $\sim 32\mu$W received at $S$, for which an efficiency around $\eta=0.11$ is available while operating with a sensitivity of $\varpi=4\mu$W \cite{Papotto.2011}\footnote{See \cite[Table~III]{Lu.2015} for summarized details on circuit performance for several RF energy harvester implementations.}.  
Moreover, $m=2$, $\sigma_d^2=-76$dBm and $k=312$ bits, while  we set $\gamma_{_{\Delta}}=10^{-3}$ to impose a high accuracy in the required value of $\gamma$ found by Algorithm~\ref{alg_1}.

\subsection{On the Accuracy of \eqref{AP2} and \eqref{AP3}}\label{res_app}
To measure the accuracy of the approximations made in \eqref{AP2} and \eqref{AP3} we evaluate the following error metric
\begin{align}
\xi=\frac{|\varepsilon_{2,\eqref{EX}}-\varepsilon_{2,(\rho)}|}{\varepsilon_{2,\eqref{EX}}},\label{errAp}
\end{align}
\begin{figure}[t!]
	\centering
	\subfigure{\label{Fig4a}\includegraphics[width=0.65\textwidth]{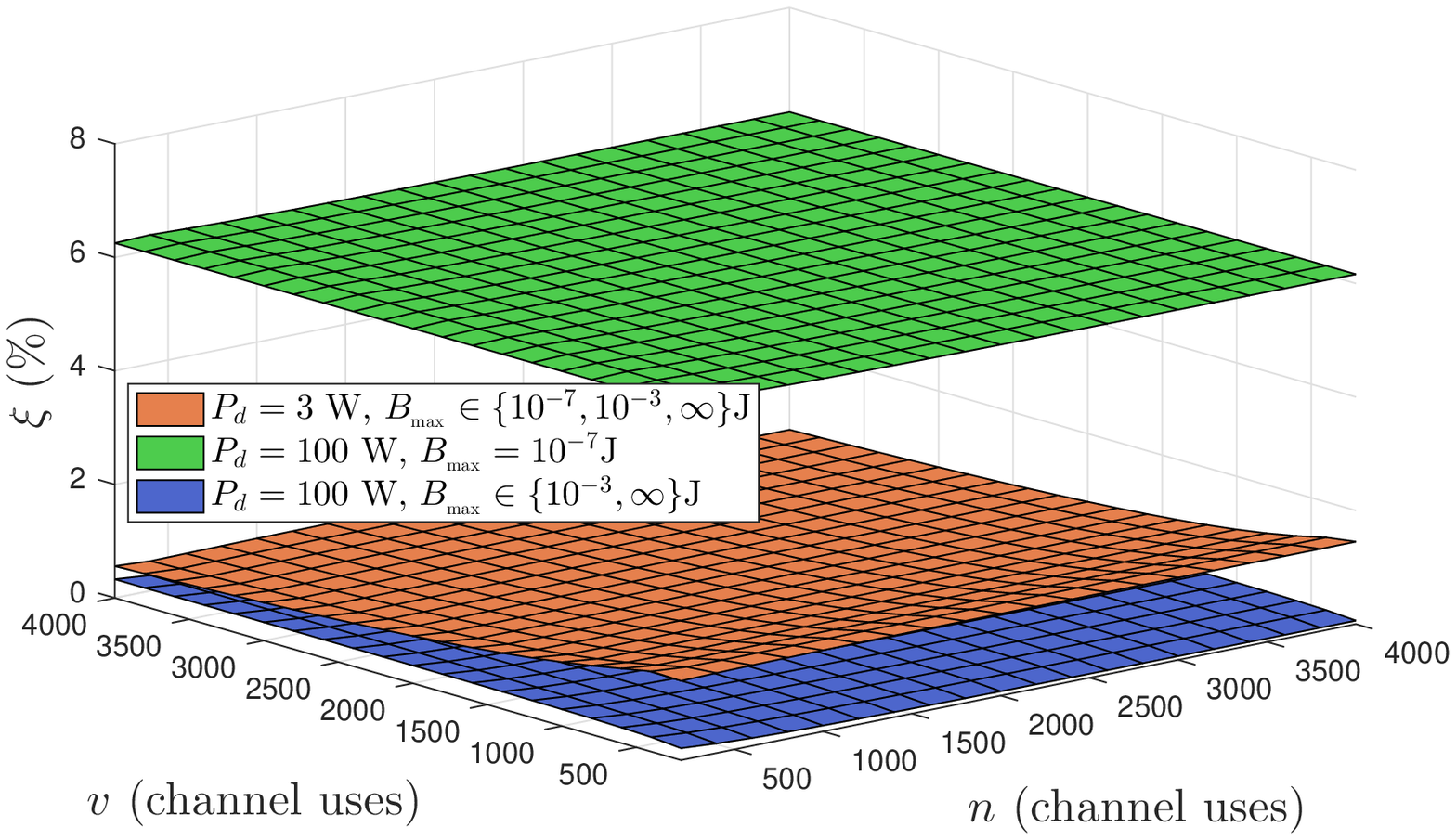}}\\
	\subfigure{\label{Fig4b}\includegraphics[width=0.65\textwidth]{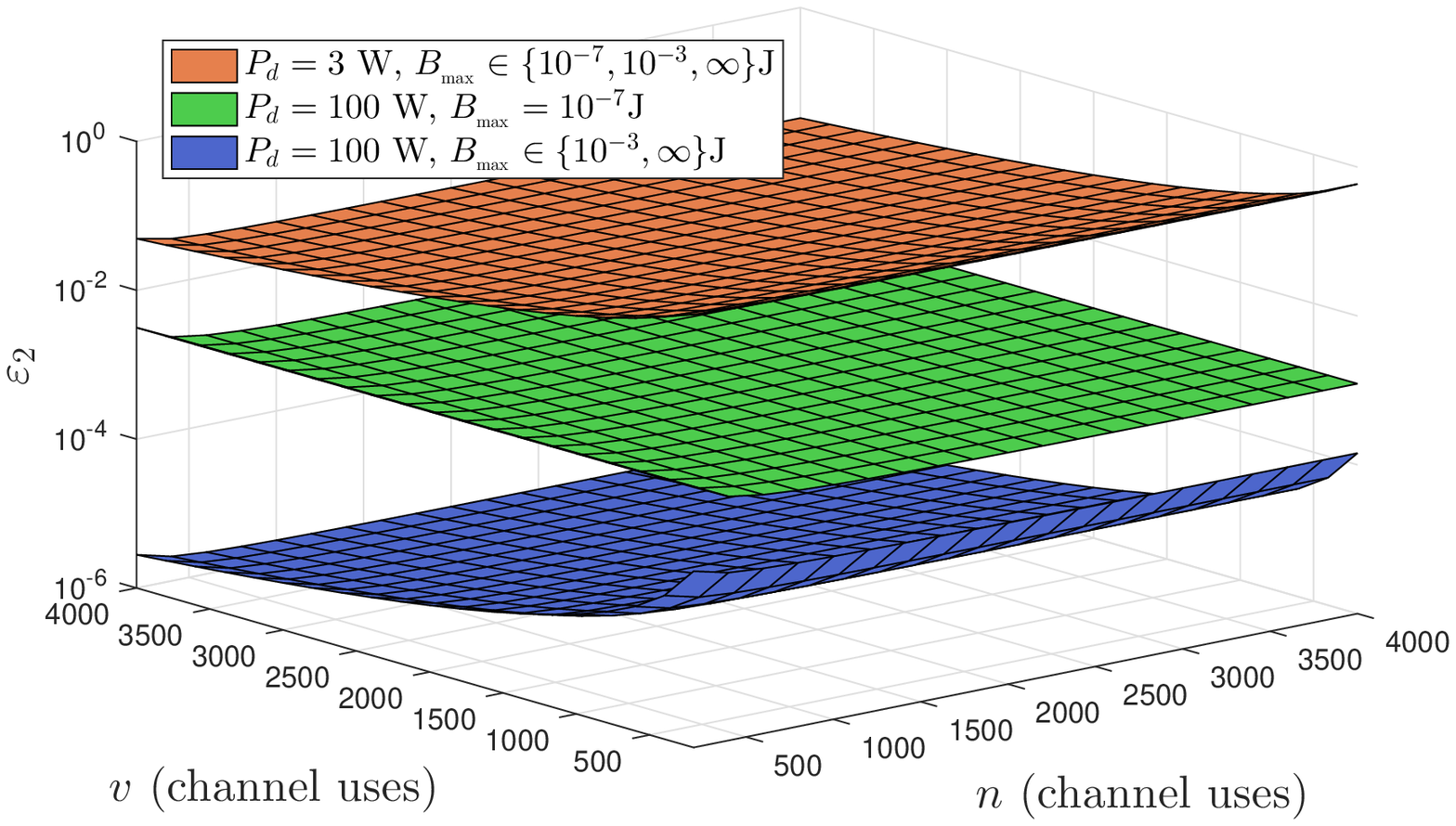}}	
	\vspace{-3mm}
	\caption{(a) $\xi(\%)$ (top) and (b) $\varepsilon_2$ (bottom), as a function of WET and WIT blocklengths. There is only one surface plot for $P_d=3$W with $B_{_{\mathrm{max}}}\in\{10^{-7},10^{-3},\infty\}$J and for $P_d=100$W with $B_{_{\mathrm{max}}}\in\{10^{-3},\infty\}$J, since they overlap and are all indistinguishable from each other.}	
	\label{Fig4}
	\vspace*{-5mm}
\end{figure} 
where $\varepsilon_{2,\eqref{EX}}$ is the error probability when communicating and given in \eqref{EX}, and $\varepsilon_{2,(\rho)}$, $\rho\in\{10,11\}$, are the approximate values given in \eqref{AP2} and \eqref{AP3}, for finite and infinite battery capacities, respectively. 
Both, $\varepsilon_{2,\eqref{EX}}$ and $\varepsilon_{2,(\rho)}$, are found for all points $(v,n)$ with $v,n\ge 100$ channel uses, where numerical evaluation is used to find $\varepsilon_{2,\eqref{EX}}$. After that, \eqref{errAp} can be computed, and according to Fig.~\ref{Fig4}a there is not a significant difference in the error approximation using \eqref{AP2} and \eqref{AP3} for a relative small transmit power $P_d=3$W, while this error starts to increase when the gap between the harvested energy and the battery size grows, e.g., $P_d=100$W and $B_{_{\mathrm{max}}}=10^{-7}$J. The exact error probability\footnote{Plotting the approximate error probability, $\varepsilon_{2,(\rho)}$, $\rho\in\{10,11\}$, would not produce appreciable differences in Fig.~\ref{Fig4}b due to the accuracy of the approximation discussed in  Fig.~\ref{Fig4}a.}, $\varepsilon_2=\varepsilon_{2,\eqref{EX}}$, for those cases is shown in Fig.~\ref{Fig4}b. Notice that for $P_d=3$W, $\varepsilon\sim10^{-1}$, while for $P_d=100$W the error when using relatively large batteries, e.g., $B_{_{\mathrm{max}}}\ge10^{-3}$J, is inferior to $10^{-4}$ due to the higher energy availability at $S$. The error probability decreases when the number of WET channel uses ($v$) increases.  
In addition, a relatively small battery, e.g., $B_{_{\mathrm{max}}}=10^{-7}$J, limits the error probability because the energy availability at each transmission round is severely limited and the chances of saving energy for future attempts decrease.
Thus, both the error approximation (Fig.~\ref{Fig4}a) and the error probability (Fig.~\ref{Fig4}b) are affected by a small battery.
\subsection{On the Convergence of Algorithm~\ref{alg_1}}\label{SecConv}
\begin{figure}[!t]
	\centering
	\subfigure{\label{Fig5a}\includegraphics[width=0.65\textwidth]{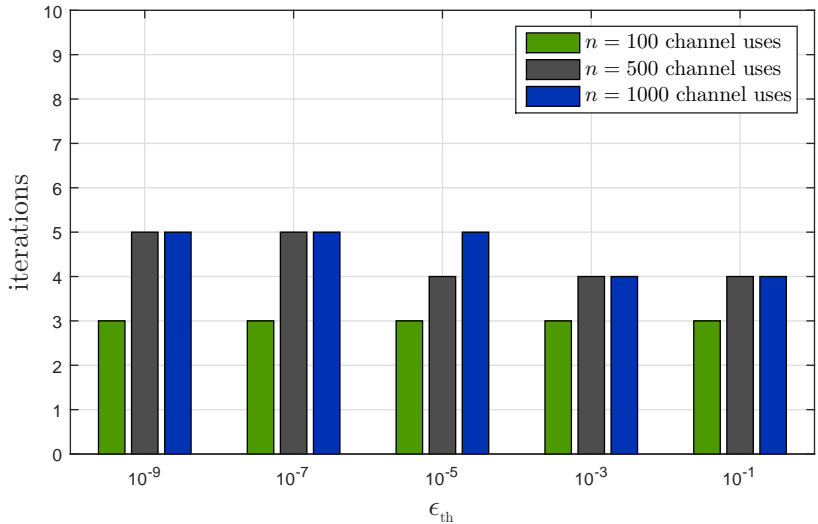}}	
	\vspace{-3mm}
	\caption{Required number of iterations to solve Algorithm~\ref{alg_1} as a function of $\epsilon_{_{\mathrm{th}}}$ for $n\in\{100,500,1000\}$ channel uses.}\label{Fig5}	
	\vspace*{-5mm}
\end{figure}
For scenarios with energy accumulation between transmission rounds, the required transmit power at $S$ is calculated at each round based on $\hat{\gamma}$, which can be found by running Algorithm~\ref{alg_1}. Fig.~\ref{Fig5} shows the required number of iterations to solve Algorithm~\ref{alg_1} as a function of the target error probability, $\epsilon_{_{\mathrm{th}}}$, for setups with $n\in\{100,500,1000\}$ channel uses. We can notice the very fast convergence of the iterative method in solving \eqref{eq}, even for a rigorous accuracy of $\gamma_{_{\Delta}}=10^{-3}$. As shown in the figure, small values of $\epsilon_{_{\mathrm{th}}}$ require more iterations, specially for relatively large values of $n$. When $n$ increases, the rate diminishes and the required $\hat{\gamma}$ becomes smaller, thus more iterations are necessary to solve the problem with the given accuracy.  If we decrease $\gamma_{_{\Delta}}$ the convergence would be slower. On the other hand, if we adopt a less demanding value of $\gamma_{_{\Delta}}$ such as $10^{-2}$, three iterations would be sufficient. Also, one of the main advantages of Algorithm~\ref{alg_1} is that it does not require an initial search interval, differently, for instance, from the bisection method. For certain search interval $I_{\gamma}$ on $\gamma$, and using the bisection method, we have that $\tfrac{I_{\gamma}}{2^{\mathrm{iterations}+1}}\le\gamma_{_{\Delta}}$. Thus, the required number of iterations shall not be less than $\log_2\Big(\tfrac{I_{\gamma}}{\gamma_{_{\Delta}}}\Big)-1$. As an example, if we search on an interval of width $I_{\gamma}=4$, for an accuracy of $\gamma_{_{\Delta}}=10^{-3}$ and $\gamma_{_{\Delta}}=10^{-2}$, 11 and 8 iterations would be respectively required under the bisection method, which are considerable higher than the required when using the Algorithm~\ref{alg_1}. As mentioned in Section~\ref{PC}, the value of $\hat{\gamma}$ has to be updated only when some element in $(n,k,\epsilon_{_{\mathrm{th}}})$ changes. Thus, it could be possible that Algorithm~\ref{alg_1} does not run in $S$, but in another entity which broadcasts its value.
\subsection{On the Performance of the Proposed Scheme}\label{res_per}
\begin{figure}[t!]
	\centering
	\subfigure{\label{Fig6a}\includegraphics[width=0.65\textwidth]{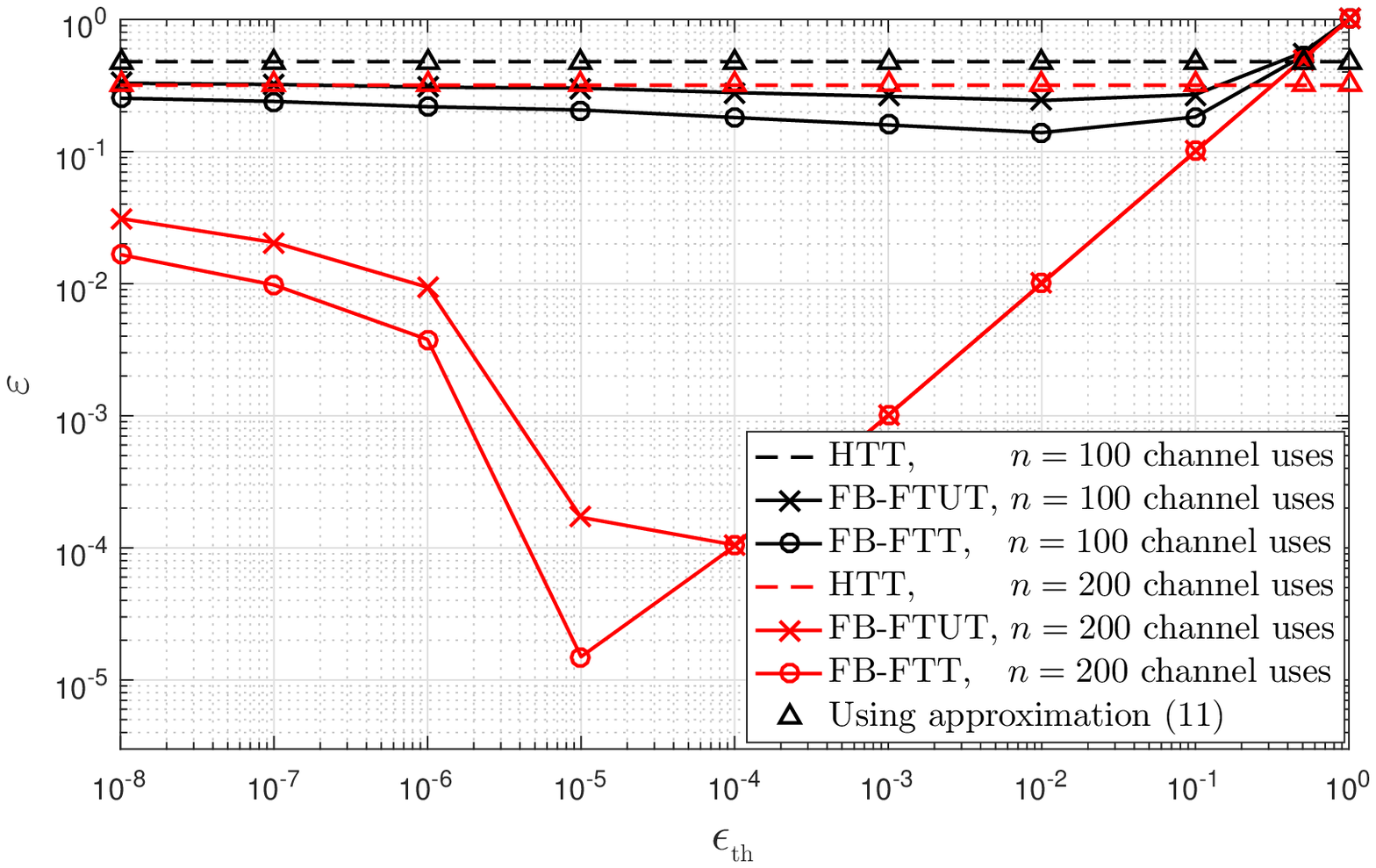}}\\
	\subfigure{\label{Fig6b}\includegraphics[width=0.65\textwidth]{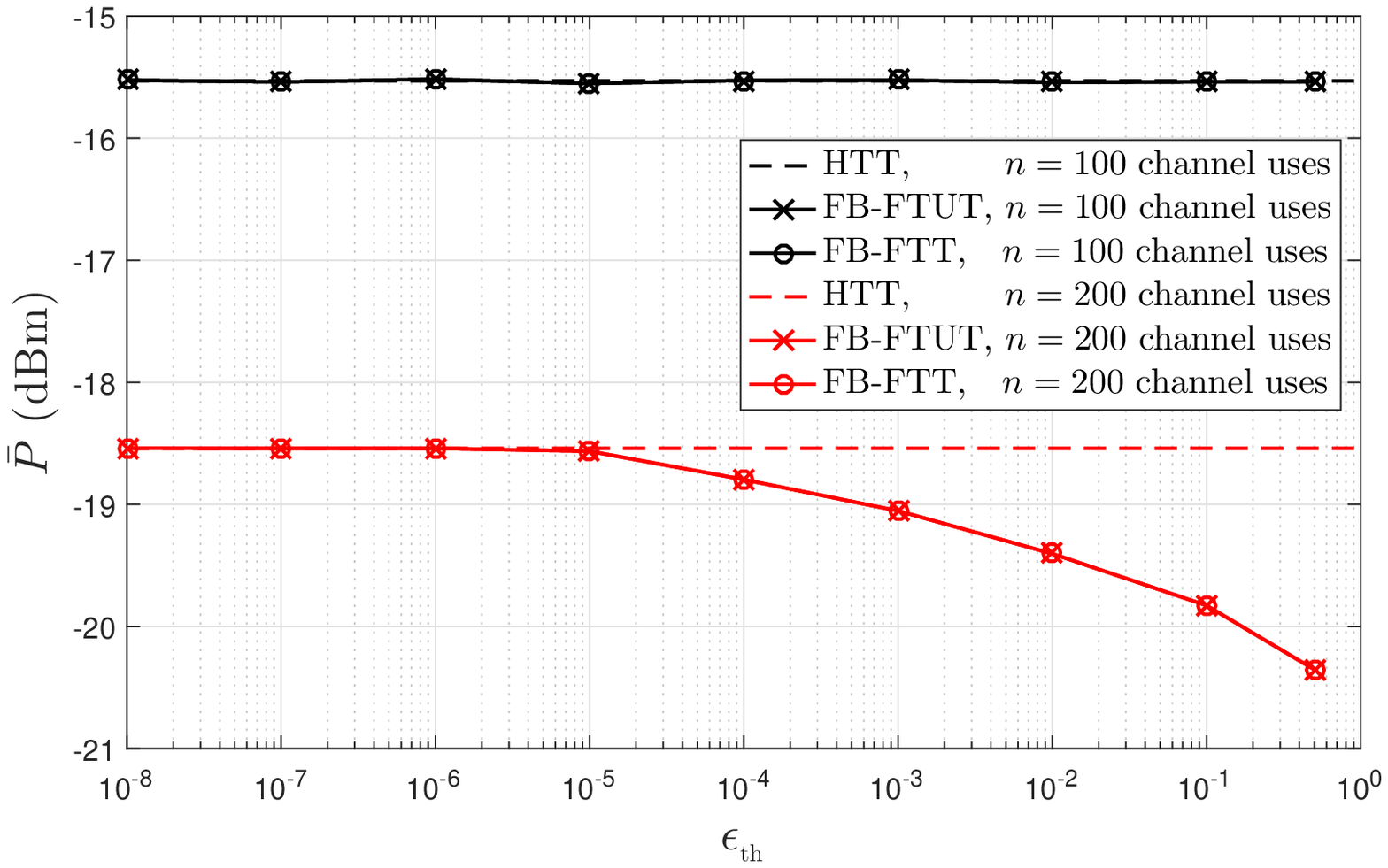}}	
	\vspace{-3mm}
	\caption{(a) $\varepsilon$ (top) and (b) $\bar{P}$ (bottom), as a function of $\epsilon_{_{\mathrm{th}}}$ for $v=800$ and $n\in\{100,200\}$ channel uses with $B_{_{\mathrm{max}}}=\infty$.}		
	\label{Fig6}
	\vspace*{-5mm}
\end{figure}
Fig.~\ref{Fig6} presents the overall error probability (Fig.~\ref{Fig6}a) and average transmit power (Fig.~\ref{Fig6}b), as a function of $\epsilon_{_{\mathrm{th}}}$  for $v=800$, $n\in\{100,200\}$ channel uses\footnote{Small $n$ and relatively large $v$ were chosen since they provide good performance according to Fig.\ref{Fig4}b.} and infinite battery capacity, while comparing the three protocols previously discussed: HTT (Section~\ref{HT}), FB-FTT and FB-FTUT (Section~\ref{PC}). In Fig~\ref{Fig6}a it is shown the existence and uniqueness of the optimum value of $\epsilon_{_{\mathrm{th}}}$ for the FB-FTT protocol\footnote{Also notice that if we draw the FB-FTT overall error performance on linear scale axes, the convexity becomes clear.}, which supports the claim made in Subsection~\ref{lemma1}. The FB-FTT scheme has the best performance for practical scenarios, e.g., $\epsilon_{_{\mathrm{th}}}<10^{-1}$, although the difference when comparing to FB-FTUT becomes smaller for relatively large values of $n$. 
For $n=200$ channel uses, the system has the best performance, thus $\epsilon_{_{\mathrm{th}}}^*$ is the smallest. In that case, the power control curves almost reach the allowable limit of $\varepsilon=10^{-5}$ for $\epsilon_{_{\mathrm{th}}}=10^{-5}$. %
Notice that when $n$ increases, the required SNR and therefore the transmit power become smaller, and even when $S$ spends more time transmitting, the energy consumption decreases as shown in Fig~\ref{Fig6}b. This holds until certain $n$, $n^*$, and beyond that the weight of the transmitting time is more relevant than the small transmit power. 
Decreasing $\epsilon_{_{\mathrm{th}}}$ allows saving more energy while the average transmit power decreases as shown in Fig.~\ref{Fig6}b, however the error performance is bounded by this value. Notice also that the average transmit power, and consequently the average energy consumption,  is practically the same for both FB-FTT and FB-FTUT strategies, thus FB-FTT is more energy efficient since it allows reaching a better error performance.
Finally, we can note the remarkable performance gap between HTT and the power control protocols, which reinforces the appropriateness of the idea behind saving energy between transmission rounds. 

\begin{figure}[t!]
	\centering
	\subfigure{\label{Fig7a}\includegraphics[width=0.65\textwidth]{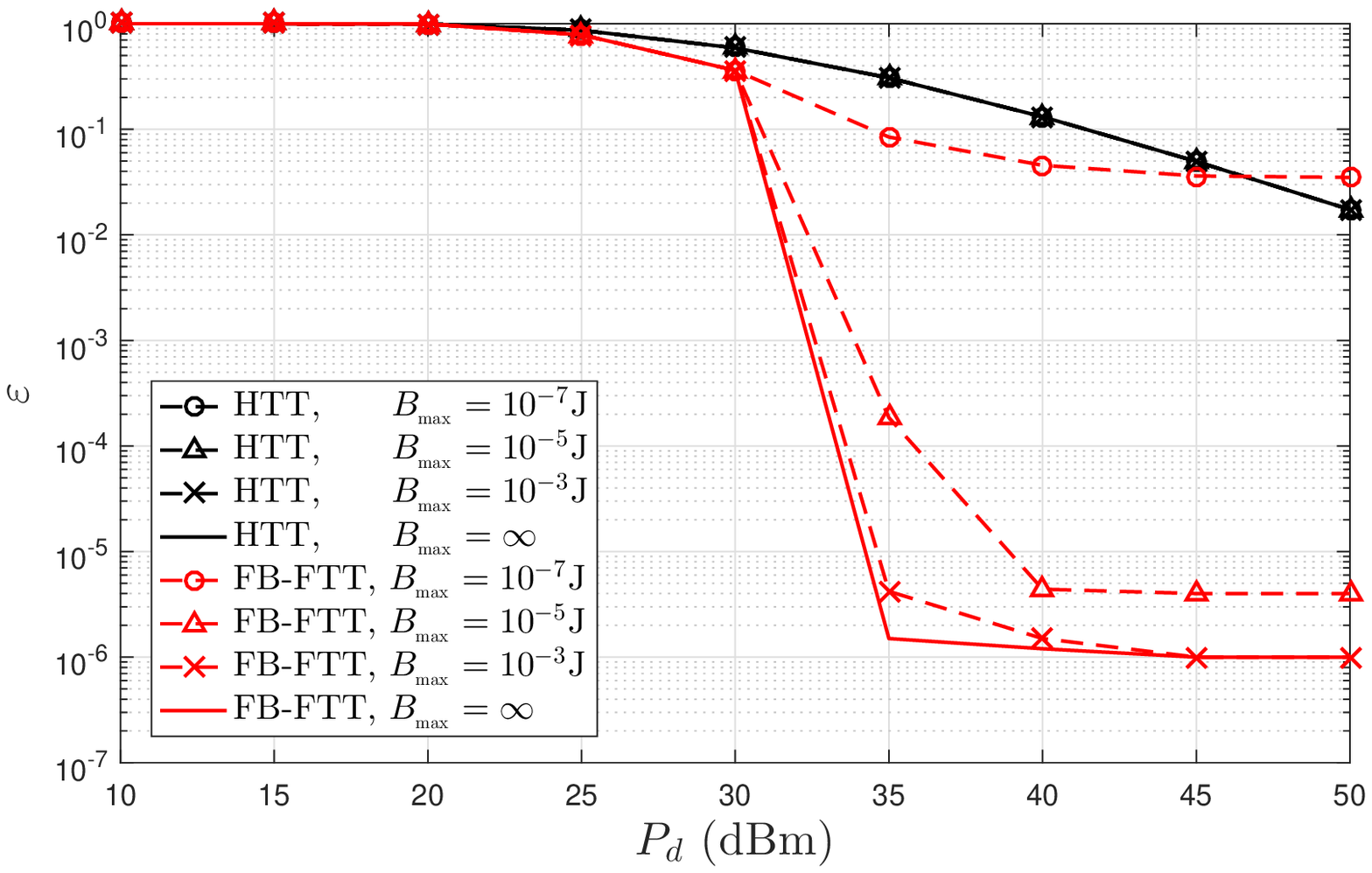}}\\
	\subfigure{\label{Fig7b}\includegraphics[width=0.65\textwidth]{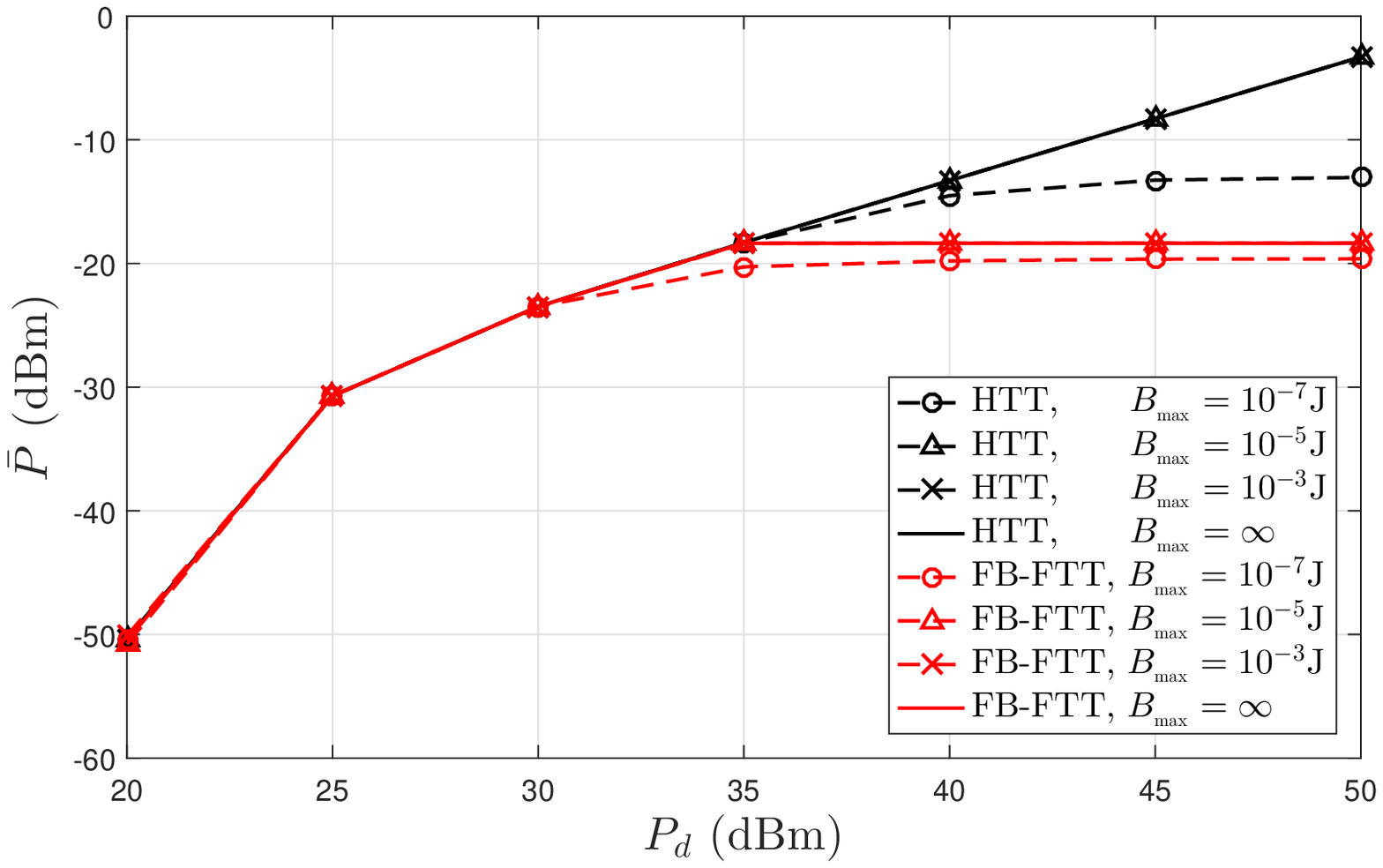}}	
	\vspace{-3mm}
	\caption{System performance as a function of $P_d$, (a) $\varepsilon$ (top) and (b) $\hat{P}$ (bottom), for $B_{_{\mathrm{max}}}\in\{10^{-7},10^{-5},10^{-3},\infty\}$J, $\epsilon_{_{\mathrm{th}}}=10^{-6}$ and $v=800$, $n=200$ channel uses.}\label{Fig7}	
	\vspace*{-5mm}	
\end{figure}
In Fig.~\ref{Fig7} we evaluate the impact of battery capacity, $B_{_{\mathrm{max}}}\in\{10^{-7},10^{-5},10^{-3},\infty\}$J, while comparing the  performance of HTT and FB-FTT protocols in terms of $\varepsilon$ (Fig.~\ref{Fig7}a) and $\hat{P}$ (Fig.~\ref{Fig7}b) as a function of $P_d$, for $\varepsilon_{_{\mathrm{th}}}=10^{-6}$ and $v=800$, $n=200$ channel uses. As shown in Fig.~\ref{Fig7}a, the impact of a finite battery capacity on the error performance is insignificant for the HTT protocol since there is no energy accumulation between transmission rounds. Therefore, only when a high amount of energy is being transferred, e.g., $P_d>50$dBm, the gap should start to be appreciable.
However, for the FB-FTT scheme the situation is more delicate since the larger the battery capacity, the greater the chances to save more energy for future transmissions, thus the better the error performance (Fig.~\ref{Fig7}a) and the larger the average energy consumption (Fig.~\ref{Fig7}b). Notice the small system performance gap between setups with $B_{_{\mathrm{max}}}=10^{-3}$J and $B_{_{\mathrm{max}}}=\infty$, and this is due to the small amount of energy being harvested in these setups with short WET phase. It is evident that for $P_d>30$dBm we have $\epsilon_{_{\mathrm{th}}}^*<10^{-6}$, since the system setup favors a better performance. In that case the gap between $B_{_{\mathrm{max}}}=10^{-3}$J and $B_{_{\mathrm{max}}}=\infty$ becomes more significant since more energy is being transfered. Also, the overall error probability improves for $P_d<30$dBm if we choose a smaller target error probability. The average transmit power for HTT protocol remains almost constant around $\hat{P}=-13$dBm for $B_{_{\mathrm{max}}}=10^{-7}$J and $P_d>40$dBm since $\mathds{E}[E_i]> B_{_{\mathrm{max}}}$, e.g., $\mathds{E}[E_i]\big|_{P_d=40\mathrm{dBm}}\approx 10^{-7}$J, thus $B_i\approx B_{_{\mathrm{max}}}$ and according to \eqref{Ps} $P_{s,i}\approx\frac{10^{-7}\mathrm{J}}{200\times 10^{-5}\mathrm{s}}=5\times 10^{-5}\rightarrow-13$dBm, both holding almost all the time. The FB-FTT protocol reaches an even small energy consumption since it does not spend all the available energy in each round, specially when channel conditions are favorable.

\begin{figure}[t!]
	\centering
	\subfigure{\includegraphics[width=0.65\textwidth]{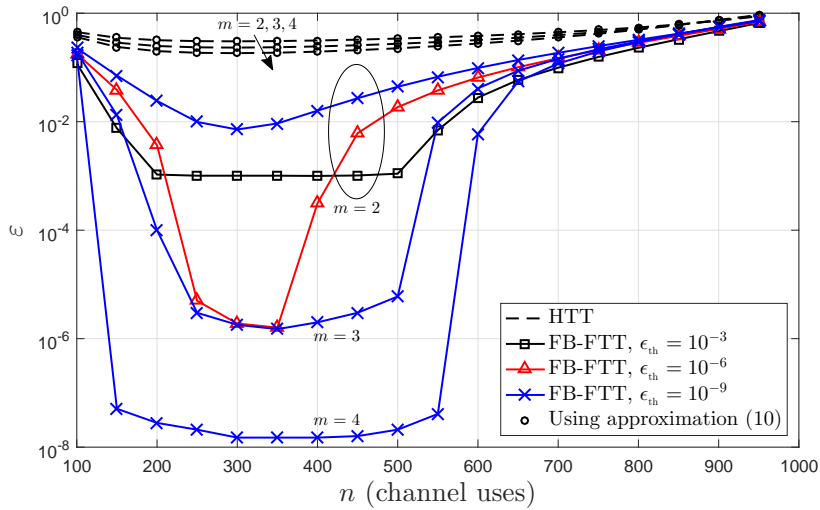}}
	\vspace{-3mm}
	\caption{$\varepsilon$ as a function of $n$ for $\epsilon_{_{\mathrm{th}}}\in\{10^{-3},10^{-6},10^{-9}\}$, $B_{_{\mathrm{max}}}=10^{-3}$J and fixed delay $n+v=1000$ channel uses. In addition to curves for $m=2$, we also plot HTT and FB-FTT with $\epsilon_{_{\mathrm{eth}}}=10^{-9}$ for $m=3,4$.}\label{Fig8}	
	\vspace*{-5mm}	
\end{figure}

In Fig.~\ref{Fig8} we fix the system delay in delivering each message by setting $n+v=1000$ channel uses, e.g., $1000T_c=10$ms, which could be fundamental in systems with very stringent delay constraints, such as Ultra-Reliable Communication over Short Term (URC-S) scenarios for future wireless systems \cite{Durisi.2015}. Results are given as a function of $n$, for $\epsilon_{_{\mathrm{th}}}\in\{10^{-3},10^{-6},10^{-9}\}$ and $B_{_{\mathrm{max}}}=10^{-3}$J. 
In URC-S scenarios, a high reliability is also required, which could be achieved via the proposed FB-FTT scheme as shown here. 
Notice that, even for very good channel conditions like $m=4$, HTT performs poorly all the time, e.g., $\varepsilon>10^{-1}$; while with the appropriate chosen value of $\epsilon_{_{\mathrm{th}}}$, the FB-FTT protocol can offer a much better performance. Focus first on the $m=2$ setup and observe that the FB-FTT protocol achieves an error probability around $\varepsilon=10^{-6}$ for $n\sim350$ channel uses. Also, when $n\ge 400$ channel uses, a target error probability greater than $10^{-6}$ is required in order to achieve the optimum system performance. Then, $n^*\sim350$, $v^*\sim650$ channel uses are approximately the optimum values for reaching the target error  $\epsilon_{_{\mathrm{th}}}^*\sim10^{-6}$ within the given delay constraint of $1000$ channel uses in channels with $m=2$. In scenarios where channels have a larger influence of the line of sight, $m$, the chances of success are greater, because it is  likely that more energy will be harvested in the downlink each time and stored for future energy-demanding transmissions\footnote{$P_d$ and $m$ impact similarly on the system performance because of that.}. The greater $m$, the smaller the optimum target error. In fact, and according to Fig.~\ref{Fig8}, it is expected that $10^{-8}<\epsilon_{_{\mathrm{th}}}^*<10^{-9}$.
All these results clearly show the convenience of a joint optimization of $n,\ v$ and $\epsilon_{_{\mathrm{th}}}$. Notice that an off-line optimization, which yields optimum parameters a priori and valid for long periods, seems more convenient for energy-constrained setups since it avoids the interchange of additional information between the nodes. 

Finally, results in Figs.~\ref{Fig6}a and \ref{Fig8} corroborate the accuracy of expressions \eqref{AP2} and \eqref{AP3}, claimed when discussing Fig.~\ref{Fig4}a.
\section{Conclusion}\label{conclusions}
In this paper, we evaluated a point-to-point communication system at finite blocklength regime with WET in the downlink, WIT in the uplink and a finite battery capacity. We attained closed-form expressions for error probability and average transmit power in scenarios where energy accumulation between transmission rounds are allowed or not, Nakagami-$m$ reciprocal channels are assumed, and the sensitivity of the energy harvester is taken into account. For scenarios allowing energy accumulation we propose a power control protocol with CSI at the transmitter side, which can be seen as a variant for finite blocklength of the FTT scheme \cite{Isikman.2016}. The numerical results show that
\begin{itemize}
	\item the closed-form approximations for the case without energy accumulation between transmission rounds (HTT protocol), under the assumption of finite and infinite battery devices, are pretty accurate when batteries are not extremely small;
	\item saving energy for future transmissions (FB-FTT scheme) allows to improve the system performance in terms of error probability while reducing the energy consumption;
	\item the proposed iterative method (Algorithm~\ref{alg_1}), which allows to find the required SNR for a target error probability, converges very fast;
	\item the optimum system performance depends on the chosen target error probability value, which in turns depends on the remaining system parameters;
	\item there is an optimum value of target error probability that minimizes the achievable error probability. However, the higher the target error probability, the lower the energy consumption;		
	\item a relatively small battery could be a limiting factor for some setups and specially for scenarios allowing energy accumulation between transmission rounds.
\end{itemize}
As a future work we intend to analyze the impact of imperfect CSI, while considering the additional delay and the energy consumption required for CSI acquisition. In addition, it could be interesting to incorporate power allocation strategies in WPCN with HARQ or/and cooperative mechanisms with finite blocklength.
\vspace*{-4mm}
\appendices 
\section{Proof of Theorem~\ref{prop_1}}\label{App_A}
Let $I_1$ and $I_2$ be the first and second integral in \eqref{EX}, respectively. Then, and accordingly to \eqref{AP}, $\mu_1=\beta$, $t=2$ for $I_1$ and $\mu_2=\beta\lambda$, $t=1$ for $I_2$. 

Substituting \eqref{AP} into \eqref{EX}, $I_1$ can be approximated as follows
\begin{align}\label{I1}
I_1&\approx\!\int_{z_{13}}^{z_{11}}\!f_G(g)\mathrm{d}g+\omega_1\!\int_{z_{14}}^{z_{12}}\!f_G(g)\mathrm{d}g-\omega_2\!\int_{z_{14}}^{z_{12}}\!g^2f_G(g)\mathrm{d}g\nonumber\\
&\stackrel{(a)}{\approx} F_G(z_{11})-F_G(z_{13})+\omega_1F_G(z_{12})-\omega_1F_G(z_{14}) -\omega_2\int_{z_{14}}^{z_{12}}\frac{m^m}{\Gamma(m)}g^{m+1}e^{-mg}\mathrm{d}g\nonumber\\
&\stackrel{(b)}{\approx}\frac{1}{\Gamma(m)}\bigg[\Gamma\big(m,mz_{13}\big)-\Gamma\big(m,mz_{11}\big)+\omega_1\Big(\Gamma\big(m,mz_{14}\big)-\Gamma\big(m,mz_{12}\big)\Big)+\nonumber\\
&\qquad\qquad\qquad+\frac{\omega_2}{m^2}\Big(\Gamma\big(m+2,mz_{12}\big)-\Gamma\big(m+2,mz_{14}\big)\Big)\bigg],
\end{align}
%
where $(a)$ comes from using the CDF definition of a random variable along with substituting the PDF of $G$ in the last term. In $(b)$, the CDF expression of $G$ is used, while the last term comes from algebraic transformations of the incomplete gamma function definition  \cite[eq.(8.2.1)]{Frank.2010}. Similarly to $I_1$, $I_2$ can be approximated as follows
\begin{align}\label{I2}
I_2&\approx\!\int_{z_{23}}^{z_{21}}\!f_G(g)\mathrm{d}g\!+\!\omega_1\!\int_{z_{21}}^{z_{22}}\!f_G(g)\mathrm{d}g\!-\!\omega_2z_{23}\!\int_{z_{21}}^{z_{22}}\!gf_G(g)\mathrm{d}g\nonumber\\
&\approx F_G(z_{21})-F_G(z_{23})+\omega_1F_G(z_{22})-\omega_1F_G(z_{21}) -\omega_2z_{23}\int_{z_{21}}^{z_{22}}\frac{m^m}{\Gamma(m)}g^{m}e^{-mg}\mathrm{d}g\nonumber\\
&\approx \frac{1}{\Gamma(m)}\bigg[(\omega_1-1)\Gamma\big(m,mz_{21}\big)+\Gamma\big(m,mz_{23}\big)-\omega_1\Gamma\big(m,mz_{22}\big)+\nonumber\\
&\qquad\qquad\qquad  +\frac{\omega_2z_{23}}{m}\Big(\Gamma\big(m+1,mz_{22}\big)-\Gamma\big(m+1,mz_{21}\big)\Big)\bigg].
\end{align}
%

Then, substituting \eqref{I1} and \eqref{I2} into $\varepsilon\approx I_1+I_2$ \eqref{EX} we attain \eqref{AP2}. Now, notice that in the case of infinite battery assumption, $\lambda\rightarrow\infty$, $\varepsilon\approx I_1$ holds. Also, $z_{11}=\zeta_1$, $z_{12}=\varphi_1$, $z_{13}=z_{15}$ and $z_{14}=z_{16}$ which allows to attain \eqref{AP3}.  \hfill 	\qedsymbol
\vspace*{-4mm}
\section{Proof of Theorem~\ref{prop_2}}\label{App_B}
By using  \eqref{Ps} and the PDF and CDF expressions of the channel gain $g$ we attain
	
	\begin{align}\label{T2_1}
	\bar{P}&=\mathds{E}[P_{s,i}]=\int_{0}^{\infty}P_{s,i}f_G(g_i)\mathrm{d}g=\frac{\eta vP_d}{n\kappa d^{\alpha}}\bigg[\int_{\varpi^*}^{\tau}gf_G(g)\mathrm{d}g+\lambda\int_{\tau}^{\infty}f_G(g)\mathrm{d}g\bigg]\nonumber\\
	&=\frac{\eta vP_d}{n\kappa d^{\alpha}}\bigg[\frac{m^m}{\Gamma(m)}\int_{\varpi^*}^{\tau}g^me^{-mg}\mathrm{d}g+\lambda\big(1-F_G(\tau)\big)\bigg]\nonumber\\
	&\stackrel{(a)}{=}\frac{\eta vP_d}{n\kappa d^{\alpha}}\bigg[-\frac{\Gamma(m+1,mg)}{\Gamma(m+1)}\bigg|_{\varpi^*}^{\tau}+\lambda\frac{\Gamma(m,m\tau)}{\Gamma(m)}\bigg],	
	\end{align}
	where the first term in $(a)$ comes from algebraic transformations of the incomplete gamma function definition \cite[eq.(8.2.1)]{Frank.2010}. Showing that \eqref{T2_1} is equivalent to \eqref{Pinf} is straightforward. Now, if $B_{_{\max}}=\infty$ then
	\begin{align}
	\bar{P}_{_{\infty}}\!=\!\mathds{E}[P_{s,i}]\!=\!\frac{\eta vP_d}{n\kappa d^{\alpha}}\int_{\varpi^*}^{\infty}\!\!gf_G(g)\mathrm{d}g\!=\!\frac{\eta vP_d}{n\kappa d^{\alpha}}\bigg[\!-\!\frac{\Gamma(m\!+\!1,mg)}{\Gamma(m\!+\!1)}\bigg]\bigg|_{\varpi^*}^{\infty}\!=\!\frac{\Gamma(m\!+\!1,m\varpi^*)}{\Gamma(m\!+\!1)}\frac{\eta vP_d}{n\kappa d^{\alpha}},
	\end{align}
	which is equal to \eqref{Pinf}.\hfill 	\qedsymbol
\vspace*{-4mm}
\section{Proof of Lemma~\ref{conv}}\label{App_C} 
    Finding $\hat{\gamma}$ reduces to solve \eqref{eq}, which  could be stated as $f(\gamma)=g(\gamma)-\gamma=0$, where $g(\gamma)=q_1q_2^{M}-1$, $M=M^{(\infty)}$ is a function of $\gamma$ \eqref{M}, $q_1=2^r\ge 1$ since $r\ge 0$, and $q_2=e^{\tfrac{Q^{-1}(\epsilon_{_{\mathrm{th}}})}{\sqrt{n}}}$. Note that    
   $f(\gamma)$ is continuous, while $f(0)=q_1-1\ge 0$ and $\lim\limits_{\gamma\rightarrow\infty}f(\gamma)=-\infty$, thus there is at least one $\gamma$ such that $f(\gamma)=0$. Based on the equation to solve, e.g., $g(\gamma)=\gamma$ with $\gamma\in\mathcal{R}^+$, we can argue as follows
   \begin{itemize}
   	\item Case I: $\epsilon_{_{\mathrm{eth}}}\ge 0.5$
   	   	
   	For this case $Q^{-1}(\epsilon_{_{\mathrm{th}}})\le 0$, thus $g(\gamma)$ is non-increasing and $\gamma$ is increasing and there is only one solution to $g(\gamma)=\gamma$.
   	
   	\item Case II: $\epsilon_{_{\mathrm{eth}}}<0.5$
   	
   	Now $Q^{-1}(\epsilon_{_{\mathrm{th}}})> 0$, thus $g(\gamma)$ is also increasing. Taking its derivatives we have
   	\begin{align}
   	g'(\gamma)&=\frac{q_1q_2^M\ln(q_2)}{(1+\gamma)^3M},\label{derq}\\
   	g''(\gamma)&=-\frac{\ln(q_2)q_1b^M}{M(1+\gamma)^2}\Big[\frac{1}{(1+\gamma)^2M}+3-\frac{\ln(q_2)}{(1+\gamma)^2}\Big],
   	\end{align}
   	where $\ln(q_2)>0$. Thus, we can claim that $g(\gamma)$ is concave if
   	\begin{align}
     \frac{1}{(1+\gamma)^2M}+3-\frac{\ln(q_2)}{(1+\gamma)^2}&> 0 \Rightarrow q_2< e^{\tfrac{1}{M}+3(1+\gamma)^2}\nonumber\\     
     Q^{-1}(\epsilon_{_{\mathrm{th}}})&<\sqrt{n}\big(\tfrac{1}{M}+3(1+\gamma)^2\big)\nonumber\\
     \epsilon_{_{\mathrm{th}}}&>Q\Big(\sqrt{n}\big(\tfrac{1}{M}+3(1+\gamma)^2\big)\Big),
   	\end{align}
   	\noindent
   	where the right side is maximized for the minimum value of $\sqrt{n}\big(\tfrac{1}{M}+3(1+\gamma)^2\big)$. Setting $n=100$, which is the minimum value for which all the analyses are valid, and $\gamma=0.1655$, which minimizes the remaining terms, we reach $\epsilon_{_{\mathrm{th}}}>Q(46.6364)\approx 4.4\times10^{-475}$. Evidently, that requirement is met for any setup of practical interest. Thus, $g(\gamma)$ is increasing and concave and since $g(0)>0$, which is the starting point of line $\gamma$, we conclude that they intersect at one point only. Therefore, the solution is unique.
   \end{itemize}
   
   Thus, we can say that the unique solution, $\hat{\gamma}$, is a fixed point of $2^{r+\tfrac{M\log_2eQ^{-1}(\epsilon_{_{\mathrm{th}}})}{\sqrt{n}}}-1$, e.g., $\hat{\gamma}=2^{r+\tfrac{M\log_2eQ^{-1}(\epsilon_{_{\mathrm{th}}})}{\sqrt{n}}}-1$ as shown in \eqref{gam}. Based on the Fixed Point Theory \cite{Agarwal.2001}, if $|g(\hat{\gamma})|<1$, the fixed point iteration in \eqref{gam} will converge to the solution.  
   Using \eqref{derq} evaluated on the solution $\hat{\gamma}$ and performing some algebraic transformations, yields
   \begin{align}
   |g'(\hat{\gamma})|&=\bigg|\frac{q_1q_2^M\ln(q_2)}{(1+\hat{\gamma})^3M}\bigg|\stackrel{(a)}{=}\frac{2^re^{\frac{MQ^{-1}(\epsilon_{_{\mathrm{eth}}})}{\sqrt{n}}}\frac{|Q^{-1}(\epsilon_{_{\mathrm{eth}}})|}{\sqrt{n}}}{(1+\hat{\gamma})^3M}\nonumber\\
   &\stackrel{(b)}{=}\frac{2^re^{\frac{\log_2(1+\hat{\gamma})-r}{\log_2e}}\frac{|\log_2(1+\hat{\gamma})-r|}{M\log_2e}}{(1+\hat{\gamma})^3M}\stackrel{(c)}{=}\frac{|\log_2(1+\hat{\gamma})-r|}{\hat{\gamma}(\hat{\gamma}+2)\log_2e}\label{gd},
   \end{align}
   where $(a)$ and $(b)$ come from using the expressions of $q_1$ and $q_2$, and $Q^{-1}(\epsilon_{_{\mathrm{eth}}})=\frac{\log_2(1+\hat{\gamma})-r}{\frac{M\log_2e}{\sqrt{n}}}$ (see \eqref{e1}), respectively; while $(c)$ is attained after substituting $M=\sqrt{1-\frac{1}{(1+\hat{\gamma})^2}}$  followed by some simplifications. Notice that for $\epsilon_{_{\mathrm{eth}}}\le0.5$, which is the case of practical interest, we have that $\log_2(1+\hat{\gamma})\ge r$, thus
   \begin{align}
   |g'(\hat{\gamma})|&<\frac{\log_2(1+\hat{\gamma})}{\hat{\gamma}(\hat{\gamma}+2)\log_2e}\le\ln 2<1,
   \end{align}
   since $\log_2(1+\hat{\gamma})\le\hat{\gamma}(\hat{\gamma}+2)$ for $\hat{\gamma}\ge 0$. Therefore, and from Banach's fixed point theorem \cite{Agarwal.2001}, the (at least) linear convergence of a Fixed-point iteration algorithm is guaranteed provided any initial point $\hat{\gamma}^{(0)}$. In this particular case, we chose $\hat{\gamma}^{(0)}=\infty\rightarrow M^{(0)}=1$. In Section~\ref{results} we show that the proposed Algorithm 1 converges very fast for practical setups.
   \hfill 	\qedsymbol
   \vspace*{-4mm}
   \section{Evidence of unique optimum value $\epsilon_{_{\mathrm{eth}}}^*$}\label{App_D}  
   Let $\epsilon_{_{\mathrm{th}}}=x\in[0,\ 1]$, $\gamma=z(x)\in\mathcal{R}^+$, $\epsilon_{_{\mathrm{out}}}=s\circ z=s(z(x))\in[0,\ 1]$, where $s(z)=\mathcal{P}[B_i<\hat{P}_{s,i}nT_c]=\mathcal{P}[g_i<\frac{z \chi }{B_i(z)}]=\frac{1}{\Gamma(m)}\Gamma\big(m,\tfrac{m\chi z}{B_i(z)}\big)$ with $\chi=\kappa d^{\alpha}\sigma_d^2nT_c$, and $\varepsilon=q(x)=(1-s(z(x)))x+s(z(x))=x+(1-x)s(z(x))\in[0,\ 1]$ according to \eqref{outage}. We know that $\epsilon_{_{\mathrm{out}}}$ is an increasing function, $s$, of $z$; however, $z$ is decreasing on $x$, thus $s$ is decreasing on $x$ as well. Also, $s(z(0))=1$ and $s(z(1))=0$, while $q(0)=q(1)=1$. Notice that when $\epsilon_{_{\mathrm{th}}}=1$, the source $S$ does not transmit and all it does is saving energy. In that case, the overall error probability is the worst possible. When $\epsilon_{_{\mathrm{th}}}$ decreases, $S$ is required to transmit with more and more power in order to fulfill the requirement. However, when $\epsilon_{_{\mathrm{th}}}=0$ the required power is practically impossible to reach and all that $S$ can do is to save energy, similarly to the case when $\epsilon_{_{\mathrm{th}}}=1$, and again the error performance is the worst possible. Evidently there is an inflexion point between $\epsilon_{_{\mathrm{th}}}=0$ and $\epsilon_{_{\mathrm{th}}}=1$, and now we aim at showing the singularity of this point.  
   
   Since every linear function is both convex and concave, we can say that  
   $x$ and $1-x$ are both convex functions. If $s(z(x))$ is convex then we could say that $q(x)$ is also convex on $x\in[0,\ 1]$ and therefore the unique minimum would be guaranteed. This is because $1-x$ and $s(z(x))$ are both convex decreasing functions, thus their product is convex \cite{Boyd.2004}, and the non-negative sum of convex functions, e.g., $(1-x)s(z(x))$ and $x$, is also convex \cite{Boyd.2004}. Let's now take a look at the second derivative of $s(z(x))$:
   \begin{align}\label{der2}
   s'(z(x))&=s'(z)z'(x)\nonumber\\
   s''(z(x))&=s''(z)z'(x)^2+s'(z)z''(x),
   \end{align}
   where $z'(x)^2> 0$, and $s'(z)>0$ since $s$ is an increasing function of $z$. We could even prove that $z''(x)>0$ for $x<0.5$, e.g., $z(x)$ is convex on the region of interest, from some analysis based on \eqref{gam} and the fact that $Q^{-1}(x)$ is convex on that region. However, and based on many and different setup simulations, since there is not an analytical expression of $s(z)$, we come to the conclusion that $s''(z)$ is not greater than $0$ in all the cases. Therefore, $s(z)$ is not always convex. Thus, it becomes intractable finding analytical arguments in order to prove that $s''(z(x))>0$ based on \eqref{der2}. In fact, our simulations show that  $s''(z(x))<0$ for few certain setups. However, even in those cases the overall error probability, $q(x)$, remains convex. Thus, the only path we can follow is by means of simulations, while exploring as many different setups as possible.
   \begin{figure}[t!]
   	\centering
   	\subfigure{\includegraphics[width=0.65\textwidth]{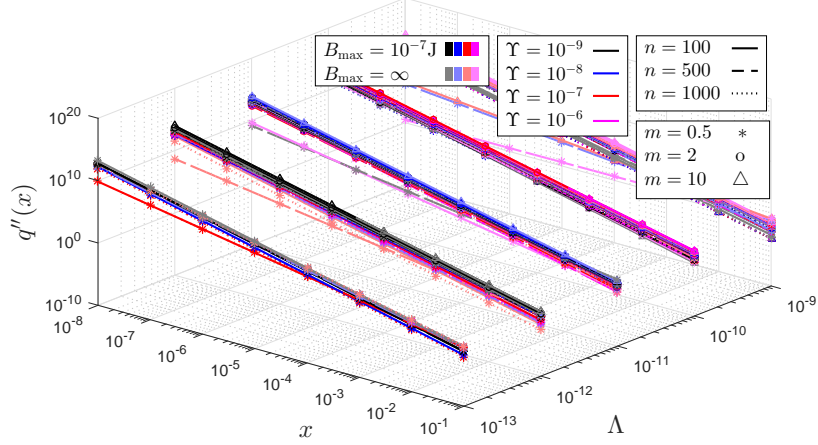}}
   	\vspace{-3mm}
   	\caption{$q''(x)$ as a function of $x$ for different values of system parameters $\Lambda$, $\Upsilon$, $B_{_{\mathrm{max}}}$, $n$, $m$, and $k=312$ bits.}	\label{Fig9}	
   	\vspace*{-5mm}
   \end{figure}
   
   Let $E_i=\Upsilon g_i\mathds{1}(g_i\ge\varpi^*)$ with $\Upsilon=\frac{\eta P_dvT_c}{\kappa d^{\alpha}}$, and $\hat{P}_{s,i}nT_c=\Lambda\frac{\hat{\gamma}n}{g_i}$ with $\Lambda=\kappa d^{\alpha}\sigma_d^2T_c$, thus, $\Upsilon$ and $\Lambda$ influence on the amount of energy being harvested and used for transmission at $S$, respectively. Also, a variation on $\eta$, $P_d$, $v$, $T_c$, $\kappa$, $d$, $\alpha$, $\sigma_d^2$, could be modeled through a variation on $\Upsilon$ and/or $\Lambda$. In Fig.~\ref{Fig9}, we show an estimated\footnote{The second derivative estimation was performed by collecting simulation data over $5\times 10^7$ channel realizations for $\epsilon_{_{\mathrm{eth}}}\in\{10^{-8},10^{-1}\}$ with spacing of $10^{-8}$ and applying numerical differentiation. Unfortunately, in order to acquire very accurate measurements, since derivative estimations are very sensitive,  we had to discard those parameter combinations leading to error probabilities below $5\times 10^{-5}$. 
   	} $q''(x)$ for the following system parameters: $k=\mathbf{312}$ bits,  $\Upsilon\in\{10^{-9},\mathbf{10^{-8}},10^{-7},10^{-6}\}$, $\Lambda\in\{10^{-13},10^{-12},\mathbf{10^{-11}},10^{-10},10^{-9}\}$, $n\in\{100,500,1000\}$ channel uses, $m\in\{0.5,\mathbf{2},10\}$, $B_{_{\mathrm{max}}}\in\{10^{-7},\infty\}$J. The values in bold are directed related with the simulation parameters in Section~\ref{results}, and notice that the impact of different message lengths, $k$, which conduces to different values of $\hat{\gamma}$, could be also modeled through variations on $\Lambda$. 
   	%
   	%
   	Notice that the convexity holds in every single case e.g., no curve presented a negative value of $q''(x)$, thus the minimum is unique in the region of interest $\epsilon_{_{\mathrm{eth}}}<0.1$.
   %
   \hfill 	\qedsymbol
   

\vspace*{-4mm}
\section{Proof of Theorem~\ref{prop_3}}\label{App_E}
The average power consumption of $S$ can be computed as 
\begin{align}
\bar{P}&=(1-\epsilon_{_{\mathrm{out}}})\mathds{E}[\hat{P}_{s,i}]=(1-\epsilon_{_{\mathrm{out}}})\int_{0}^{\lambda}\hat{P}_{s,i}f_G(g)\mathrm{d}g\stackrel{(a)}{=}(1-\epsilon_{_{\mathrm{out}}})\gamma\kappa d^{\alpha}\sigma_d^2\int_{0}^{\lambda}\frac{m^m}{\Gamma(m)}g^{m-2}e^{-mg}\mathrm{d}g\nonumber\\
&\stackrel{(b)}{=}-(1-\epsilon_{_{\mathrm{out}}})\gamma\kappa d^{\alpha}\sigma_d^2m\frac{\Gamma(m-1,mg)}{\Gamma(m)}\bigg|_0^{\lambda},\label{PC3}
\end{align}
where $(a)$ comes from using \eqref{Psi}, and  $(b)$ from algebraic transformations of the incomplete gamma function definition \cite[eq.(8.2.1)]{Frank.2010} with $m>1$. We attain \eqref{PC1} straightforward from \eqref{PC3}. Now, substituting $\lambda=\infty$ into \eqref{PC3} yields
\begin{align}
\bar{P}_{_{\infty}}&=-(1-\epsilon_{_{\mathrm{out}}})\gamma\kappa d^{\alpha}\sigma_d^2m\frac{\Gamma(m-1,mg)}{\Gamma(m)}\bigg|_0^{\infty}=(1-\epsilon_{_{\mathrm{out}}})\gamma\kappa d^{\alpha}\sigma_d^2m\frac{\Gamma(m-1,0)}{\Gamma(m)},	
\end{align}
which is equal to \eqref{PC2}. \hfill 	\qedsymbol

\vspace*{-4mm}
\bibliographystyle{IEEEtran}
\bibliography{IEEEabrv,references}

\end{document}